\documentclass[11pt, letterpaper]{article}
\usepackage[utf8]{inputenc}
\usepackage[margin=1in]{geometry}
\usepackage{amssymb}
\usepackage{amsthm}
\usepackage{amsmath}
\usepackage{mathtools}
\usepackage{bbm}
\usepackage{xcolor}
\usepackage{algpseudocode}
\usepackage{subcaption}
\usepackage[nodayofweek,level]{datetime}
\usepackage{graphicx}
\usepackage{hyperref}
\usepackage[capitalise]{cleveref}
\usepackage{multicol}

\newcommand{\SymmetricPoissonTower}{\textsf{SymmetricPoissonTower}}

\title{Harmonic Decomposition in Data Sketches\thanks{This work was supported by NSF Grant CCF-2221980.}}

\author{Dingyu Wang\\wangdy@umich.edu\\University of Michigan}
\date{}

\usepackage{amssymb}
\usepackage{amsthm}
\usepackage{amsmath}
\usepackage{mathtools}
\usepackage{bbm}
\usepackage{wasysym}

\theoremstyle{plain}
\newtheorem{theorem}{Theorem}
\newtheorem{lemma}{Lemma}

\newtheorem{corollary}{Corollary}
\newtheorem{conjecture}{Conjecture}

\theoremstyle{definition}
\newtheorem{definition}{Definition}

\theoremstyle{remark}
\newtheorem{remark}{Remark}

\DeclarePairedDelimiter{\floor}{\lfloor}{\rfloor}

\DeclarePairedDelimiter{\norm}{\|}{\|}

\newcommand{\ind}[1]{\mathbbm{1}\left[{#1}\right]}
\newcommand{\var}{\mathbb{V}}
\newcommand{\Cov}{\mathrm{Cov}}
\newcommand{\pr}{\mathbb{P}}
\newcommand{\E}{\mathbb{E}}

\newcommand{\Update}{\mathsf{Update}}
\newcommand{\supp}{\operatorname{supp}}
\newcommand{\poly}{\operatorname{poly}}

\newcommand{\Poisson}{\operatorname{Poisson}}
\newcommand{\SymmetricPoisson}{\operatorname{SymmetricPoisson}}

\newcommand{\R}{\mathbb{R}}

\newcommand{\Z}{\mathbb{Z}}
\newcommand{\C}{\mathbb{C}}
\newcommand{\N}{\mathbb{N}}

\renewcommand{\Re}{\mathrm{Re}}
\renewcommand{\Im}{\mathrm{Im}}

\begin{document}
\maketitle

\begin{abstract}
In the turnstile streaming model, a dynamic vector 
$\mathbf{x}=(\mathbf{x}_1,\ldots,\mathbf{x}_n)\in \Z^n$ is updated by a stream of entry-wise increments/decrements. 
Let $f\colon\Z\to \R_+$ be a symmetric function with $f(0)=0$. 
The \emph{$f$-moment} of $\mathbf{x}$ is defined to be

\[
f(\mathbf{x}) := \sum_{v\in[n]}f(\mathbf{x}_v).
\]

We revisit the problem of constructing a \emph{universal sketch} that can estimate many different $f$-moments. Previous constructions of universal sketches rely on the technique of sampling with respect to the $L_0$-mass (uniform samples) or $L_2$-mass ($L_2$-heavy-hitters), whose universality comes from being able to evaluate the function $f$ over the samples. To get samples, hash collisions are deliberately detected and avoided (with high probability), e.g.~singleton-detectors are used in $L_0$-sampling and the CountSketch is used in $L_2$-sampling. Such auxiliary data structures introduce significant overhead in space. Apart from this issue, sampling-based methods are shown to perform poorly for estimating certain ``nearly periodic functions'' where $\Omega(\poly(n))$ samples are needed.

\medskip 

In this paper, we propose a new universal sketching scheme that is almost ``dual'' to the sampling-based methods. Instead of evaluating $f$ on samples, we decompose $f$ into a linear combination of \emph{homomorphisms} $f_1,f_2,\ldots$ from $(\Z,+)$ to $(\C,\times)$, where the $f_k$-moments can be estimated regardless of hash collisions---\emph{because each $f_k$ is a homomorphism!} Then we synthesize the estimates of the $f_k$-moments to obtain an estimate of the 
$f$-moment.  Universality now comes from the fact
that we can weight the $f_k$-moments arbitrarily, 
where the correct weighting depends on 
the \emph{harmonic structure} of the function $f$. 

\medskip 

In contrast to the sampling-based methods, the new \SymmetricPoissonTower{} sketch takes the harmonic approach. It \emph{embraces} hash collisions instead of avoiding them, which saves multiple $\log n$ factors 
in space, e.g., when estimating all $L_p$-moments ($f(z) = |z|^p,p\in[0,2]$). For many nearly periodic functions, the \SymmetricPoissonTower{} is \emph{exponentially} more efficient than sampling-based methods. We conjecture that the \SymmetricPoissonTower{} is \emph{the} universal sketch that can estimate every tractable function $f$.

\end{abstract}

\section{Introduction}

In the turnstile streaming model (\cite{AlonMS99,cormode2003comparing}), 
a \emph{frequency vector} $\mathbf{x}=(\mathbf{x}_1,\ldots,\mathbf{x}_n)\in \Z^n$ is initialized as 
$0^n$ and updated through a stream of pairs $(v,\Delta)$, 
$v\in[n] = \{1,2,\ldots,n\}$ and $\Delta\in \Z$.
Upon receiving $(v,\Delta)$, $\mathbf{x}_v \gets \mathbf{x}_v + \Delta$.
We consider the problem of estimating the \emph{$f$-moment} $f(\mathbf{x})$, where $f\colon\Z\to\R_+$ is symmetric ($f(z)=f(-z)$) with $f(0)=0$, and
\begin{align*}
    f(\mathbf{x}) \coloneqq \sum_{v\in[n]}f(\mathbf{x}_v).\tag{$f$-moment of stream $\mathbf{x}$}
\end{align*}
Let $M$ be the entry-wise frequency bound.
We say an estimator $V$ \emph{$(1\pm\epsilon)$-approximates} the $f$-moment, if for 
any $\mathbf{x} \in [-M,M]^n$, 
\begin{align*}
    \pr\left((1-\epsilon)f(\mathbf{x})\leq |V-f(\mathbf{x})|\leq (1+\epsilon)f(\mathbf{x})\right) \geq 2/3.
\end{align*}
 To simplify the presentation, we assume the usual regime where $\epsilon^{-1}\leq O(\poly(n))$ and $M\leq O(\poly(n))$. We define two important notions---\emph{tractability} and \emph{universality}---regarding the $f$-moment estimation problem.
\begin{description}
    \item[Tractability of functions] We say a function $f$ is \emph{tractable} \cite{braverman2010zero} if for any fixed $\epsilon>0$, there exists a 
$(1\pm \epsilon)$-approximation algorithm of the $f$-moment 
using $O_\epsilon(\mathrm{polylog}(n))$ space. A central problem in the streaming literature is to characterize the class of \emph{tractable} functions \cite{AlonMS99,braverman2010zero,braverman2016streaming}.\footnote{In this work we use the tractability notion in \cite{braverman2010zero} where a function is tractable if it can be approximated with polylog-space. A slightly different definition is used in  \cite{braverman2016streaming} where functions are defined to be tractable if the space is sub-polynomial ($n^{o(1)}$). } 
    \item[Universality of sketches] We say a sketch is \emph{$\mathcal{C}$-universal}, where $\mathcal{C}$ is a class of functions, if for any $f\in \mathcal{C}$, a $(1\pm\epsilon)$-approximation of the $f$-moment can be returned by the sketch. It is highly desirable to find sketches that are $\mathcal{C}$-universal for as large a $\mathcal{C}$ as possible \cite{braverman2015universal,liu2016one}.
\end{description}


The \emph{ultimate} goal for $f$-moment sketching is thus to find a ``truly universal'' sketch that can $(1\pm\epsilon)$-approximate \emph{all tractable $f$-moments} with optimal log factors in the space complexity. A prerequisite of this goal is a complete characterization of tractable functions.

\medskip

The tractability of classic $L_p$-moments ($f$-moment with $f(z)=|z|^p$, $p\geq 0$)\footnote{We define $|z|^0 = \ind{z\neq 0}$.} is well-understood. On the one hand, $L_p$-moments with $p\in[0,2]$ can be $(1\pm \epsilon)$-approximated with $O(\epsilon^{-2}\log(n))$ bits of space \cite{AlonMS99,indyk2006stable,cormode2003comparing,KaneNW10}.\footnote{For $L_0$, only $O(\epsilon^{-2}\log n (\log \epsilon^{-1}+\log\log (n M)))$ bits are needed \cite{KaneNW10}.} On the other hand, Bar-Yossef, Jayram, Kumar, and Sivakumar~\cite{bar2004information} 
    proved that estimating the $L_{p}$-moment requires $\Omega(n^{1-2/p+o(1)})$ bits, which is $\Omega(\poly(n))$ space for any $p>2$. Therefore, $L_p$ is tractable if and only if $p\in[0,2]$. Nevertheless, classic $L_p$-sketches ($p\in[0,2]$) are all \emph{single-purposed} where the sketch design depends crucially on the parameter $p$. Thus despite being very efficient in space (only one log factor), classic $L_p$-sketches are extremely \emph{non-universal}.

To understand the tractability of other functions, the strategy has been to fully explore the class of all approximable functions by some ``powerful'' sketch $\mathcal{S}$ \cite{braverman2010zero,braverman2015universal,braverman2016streaming} and then hopefully prove any function that is not approximable by $\mathcal{S}$ cannot be approximated by any other polylog-space sketch either. Note that this strategy actually pursues tractability and universality at the same time, since once a function class $\mathcal{C}$ is proved to be approximable by the sketch $\mathcal{S}$, $\mathcal{S}$ is automatically a $\mathcal{C}$-universal sketch. We now discuss the existing techniques.

\paragraph{$L_0$-sampling.} A simple universal sketch is to sample $m$ elements $\mathbf{x}_{v_1},\ldots,\mathbf{x}_{v_m}$ uniformly from $\supp(\mathbf{x})$ and then estimate the $f$-moment by the empirical mean $V_f=\frac{\norm{\mathbf{x}}_0}{m}\sum_{j\in[m]}f(\mathbf{x}_{v_j})$. The sampling can be done 
with $L_0$-samplers (see~\cite{CormodeF14}, $O(\log^2n)$ space to get one $L_0$-sample) and the support size $\norm{\mathbf{x}}_0$ can be estimated by $L_0$-sketches~\cite{cormode2003comparing,KaneNW10}. 
This scheme is natural but the estimates are poor if $f$ varies a lot. Specifically, for functions $f$ such that $f(x)>0$ when $x\neq 0$, the number $m$ of samples  needs to be $\Omega(\epsilon^{-2}\max_{j\in[M]}f(j)/\min_{j\in [M]}f(j))$ for $V_f$ to be a $(1\pm \epsilon)$-approximation of the $f$-moment; see, Chestnut~\cite{chestnut2015stream}. 
For functions decreasing on $[1,M]$, it is proved in \cite{braverman2015universal} that all tractable moments can be estimated using this scheme. However, functions like $f(x)=|x|^p$ for $p\in(0,2]$ will take $\Omega(M^p)$ number of $L_0$-samples to approximate.


\paragraph{$L_2$-heavy hitters.} A more powerful sketch, used in~\cite{braverman2010zero,braverman2013generalizing,braverman2016streaming}, is to use variations on Indyk and Woodruff's~\cite{IndykW03} method of collecting 
$L_2$-heavy hitters using CountSketch \cite{CharikarCF04} at different subsampling levels,
and estimate $f$-moments based on the $f$-values of the ``recursive'' heavy hitters.
This technique is powerful enough to estimate all tractable $f$-moments, except for a class of \emph{nearly periodic} 
functions~\cite{braverman2016streaming}. 
One such function defined in \cite[\S 5]{braverman2016streaming} is $g_{np}\colon \Z \to \R_+$.\footnote{In other words, $\tau(x)$ is the least significant bit position in the binary representation of $x>0$.}
\begin{align}
g_{np}(0)=0 \;\mbox{ and }\; g_{np}(x)=2^{-\tau(|x|)}, \;\mbox{ where }\; \tau(x) = \max\{j\in \N : 2^j|x\}.      \label{eq:gnp}
\end{align}
For functions like $g_{np}$, the entries with large $g_{np}$-values are not necessarily $L_2$-heavy hitters, so sampling $L_2$-heavy hitters
are generally not helpful.
However, 
Braverman, Chestnut, Woodruff, and Yang~\cite[\S 5]{braverman2016streaming} gave an 
\emph{ad hoc} sketch for approximating $g_{np}$-moments 
with $O(\epsilon^{-8}\log^{14}n\log M)$ space, showing that $g_{np}$ is indeed tractable. A recent work by Pettie and Wang \cite{pettie2024sketching} further characterizes a class of tractable nearly periodic functions but their sketch is  single-purposed, with different $f$ requiring different sketches. 

\medskip
We summarize the two main drawbacks of existing universal sketches as follows.
\begin{description}
    \item[Drawback 1: Overhead for sampling] While a single $L_0$-sampler is quite space-efficient ($O(\log^2 n)$ bits), the class of approximable functions using polylog number of uniform samples is very limited (e.g., it fails to estimate any $L_p$-moment, $p>0$). The more powerful $L_2$-heavy hitter based 
sketch in \cite{braverman2016streaming} handles more function moments but it 
occupies siginificantly more space ($O(h(M)\epsilon^{-2}\log^6 n\log\log n)$ bits\footnote{$h$ is a sub-polynomial 
function so that for any $0<x<y$, $f(y)/f(x)\in[1/h(y),(y/x)^2 h(y)]$. See \cite{braverman2016streaming} for further details.}) as CountSketch is used on every subsampling level. The natural question is
\begin{center}
    \emph{Are such auxiliary sampling-related data structures necessary?}
\end{center}
    \item[Drawback 2: Fundamental mismatch with nearly periodic functions] The polylog-space overhead for sampling is non-ideal in practice, though it can be ignored when characterizing tractability. A more fundamental limitation is that, even if $L_0$ or $L_2$ samples are given for free, the sampling based method will use \emph{exponentially} more space when approximating certain nearly periodic functions, in comparison to the best possible sketch. Essentially, sampling based methods will need $\poly(n)$ samples if the function $f$ ever drops polynomially ($f(x)<f(1)/\poly(n)$ for some $x\in[-M,M]$). Even though most ``polynomially dropping'' functions are proved to be not tractable~\cite{braverman2015universal,braverman2016streaming}, Braverman Braverman, Chestnut, Woodruff, and Yang~\cite{braverman2016streaming} realize some of them can be estimated (e.g., $g_{np}$) using polylog-space by other non-sampling based methods. It is thus intriguing to ask
    \begin{center}
        \emph{Is there a universal sketch that handles tractable nearly periodic functions as well?}
    \end{center}
\end{description}

The discussion above suggests that one must depart from sampling based methods to remove the overhead for sampling-related subroutines, and to further improve the universality of sketches.

\paragraph{Our Contribution: A New \emph{Harmonic} Approach for Universal Sketching} Our main contribution is a whole new sketching scheme that does not use \emph{any} explicit samples but relies on the \emph{harmonic structure} of the target function $f$. Without any auxiliary data structure for sampling, the new  \SymmetricPoissonTower{} sketch is very space-efficient, occupying only $O(\epsilon^{-2}\log^2 n)$ bits of space. It can $(1\pm \epsilon)$-approximate all the natural functions including the common $L_p$-moments ($f(x)=|x|^p$, $p\in[0,2]$), logarithmic-moment ($f(x)=\log(1+|x|)$), and ``soft cap''-moment ($f(x)=1-e^{-|x|}$). Note that $L_0$-sampling based methods with matching space can only estimate $f$-moment with $\max_{j\in[M]}f(j)/\min_{j\in [M]}f(j)=O(1)$. Thus among the examples, the $L_0$-sampling method~\cite{braverman2015universal,chestnut2015stream} can only estimate the ``soft cap''-moment but fails to estimate $L_p$-moments and the logarithmic-moment.\footnote{Assuming $M=O(\poly(n))$, one needs $O(\epsilon^{-2}\log M)$ $L_0$-samples to estimate the logarithmic-moment. Thus the total space needed is $O(\epsilon^{-2}\log^3 n)$ which is one log factor more than the space used by the harmonic approach. } The powerful $L_2$-sampling framework~\cite{braverman2016streaming} 
can handle all the examples above, but it requires $\Omega(\epsilon^{-2}\log^6 n \log\log n)$ bits. 

More than saving polylog factors, 
the new harmonic approach also \emph{natively} tracks a large class of nearly periodic functions, including $g_{np}$, for which the sampling based methods provably need polynomially many samples. Thus the new approach also makes significant progress towards the ultimate goal of finding a ``truly universal'' sketching scheme. See \cref{fig:summary}.

Apart from the technical contributions, the new  \SymmetricPoissonTower{} sketch is also conceptually novel. In contrast to the previous methods which \emph{avoid} hash collisions,\footnote{For $L_0$-sampling, hash collisions refer to the usual notion of multiple elements hashed to one cell. For $L_2$-sampling, hash collisions refer to two or more elements with significant $L_2$-mass hashed into one cell (there can be many insignificant elements hashed to it which serve as ``background noise'' in  CountSketch \cite{CharikarCF04}). } the new harmonic approach \emph{embraces} hash collisions. 
The key new ingredient is an insight that enables us to infer information about $f$-moments from a cell with an \emph{unknown} number of indices hashed to it. 

\subsection{Insight: How to Estimate $f$-Moments under Hash Collisions?}\label{sec:insight}
Imagine hashing every element in $[n]$ to a cell, independently with probability $p=2^{-k}$. Let $J\subset[n]$ be the set of indices hashed to the cell 
and let $X = \sum_{v\in J} \mathbf{x}_v$ be stored in the cell. 
When $k$ and $\lambda=|\supp(\mathbf{x})|$ are large,
$|J|$ tends to a Poisson distribution.  For simplicity, we shall \emph{Poissonize} the cell\footnote{Poissoinzation is a standard technique in subsampling which simulates each update with a $\Poisson(1)$ number of updates. It simplifies the analysis without affecting either update time and space. See, e.g., \cite{chen2024space}.} where $|J|\sim \Poisson(p\lambda)$ indices $\{I_j\}_{1\leq j\leq |J|}$ are selected uniformly at random from $\supp(\mathbf{x})$ \emph{with replacement}, so $X$ is precisely
\begin{align}
    X = \sum_{j=1}^{\mathrm{Poisson}(p\lambda)} \mathbf{x}_{I_j}, \label{eq:X-poisson}
\end{align}
where the $I_j$s are i.i.d.~copies of $I\sim \mathrm{Uniform}(\supp(\mathbf{x}))$. 
The goal is to estimate the $f$-moment $f(\mathbf{x})=\sum_{v\in[n]}f(\mathbf{x}_v)$ from the observation $X$. 
We begin with the trivial but important observation that when $f$ is \emph{linear}, $X$ betrays lots of useful information about $f(\mathbf{x})$.  In particular,
\begin{align}
    \E f(X) &= \E f\left(\sum_{j=1}^{\mathrm{Poisson}(p\lambda)} \mathbf{x}_{I_j}\right)  & \text{By \cref{eq:X-poisson}}\nonumber\\
    &= \E \sum_{j=1}^{\mathrm{Poisson}(p\lambda)} f\left( \mathbf{x}_{I_j}\right) & \text{From the linearity of $f$}\tag{$\star$}\\
    &= p\lambda \E f(\mathbf{x}_I) & \text{Linearity of expectation}\nonumber\\
    &= p f(\mathbf{x}) & \text{$f(\mathbf{x})=\lambda\E f(\mathbf{x}_I) $} \nonumber
\end{align}
Thus $f(X)/p$ is an unbiased estimator for $\lambda \E f(\mathbf{x}_I)$.

\paragraph{The key insight.}
In the example above it is not important that $f$ 
is linear \emph{per se}, but that
it is a \emph{homomorphism} (allowing the step $(\star)$), in this case from 
$(\Z,+)$ to $(\C,+)$.  
Moreover, other homomorphisms are equally useful.
For example, suppose that $f$ is instead a homomorphism from $(\Z,+)$ to $(\C,\times)$, the \emph{multiplicative} group of 
complex numbers, i.e., for any $x,y\in \Z$, $f(x+y) = f(x)f(y)$.
We can still infer information about $\E f(\mathbf{x}_I)$ from $f(X)$, even when $X$ is the sum of more than one coordinate of $\mathbf{x}$, since $f(X)$ is just the product of $f(\mathbf{x}_{I_j})$s. We compute $\E f(X)$ as follows.
\begin{align}
    \E f(X) &= \E f\left(\sum_{j=1}^{\mathrm{Poisson}(p\lambda)} \mathbf{x}_{I_j}\right) & \text{By \cref{eq:X-poisson}}\nonumber\\
    &= \E \prod_{j=1}^{\mathrm{Poisson}(p\lambda)} f\left( \mathbf{x}_{I_j}\right) & \text{Since $f:(\Z,+)\to (\C,\times)$ is a homomorphism}\tag{$\star\star$}\\
    &= e^{-p\lambda}\sum_{k=0}^\infty \frac{(p\lambda)^k}{k!} \E \prod_{j=1}^{k} f(\mathbf{x}_{I_j}) & \text{Definition of $\mathrm{Poisson}(p\lambda)$}\nonumber\\
    &=e^{-p\lambda}\sum_{k=0}^\infty \frac{(p\lambda)^k}{k!} (\E f(\mathbf{x}_{I}))^k & \text{Since the $\mathbf{x}_{I_j}$s are i.i.d.}\nonumber\\
    &=e^{p\lambda (\E f(\mathbf{x}_{I})-1)} & \text{Taylor expansion of $\exp(p\lambda\E f(\mathbf{x}_I))$}\nonumber\\
    &=e^{pf(\mathbf{x})} & \text{$f(\mathbf{x})=\lambda\E (f(\mathbf{x}_I)-1) $}\nonumber
\end{align}
Note that since a multiplicative homomorphism maps $0\in(\Z,+)$ to $1\in (\C,\times)$, we must calibrate $f(z)$ to $f^\dagger(z)=f(z)-1$ when computing the $f$-moment so that $f^\dagger(0)=0$. Therefore we have $f(\mathbf{x}) \coloneqq \sum_{v\in[n]} (f(\mathbf{x}_v)-1)=\lambda \E (f(\mathbf{x}_I)-1)$ in the last step.\footnote{In general, one may define the $f$-moment to be $\sum_{v\in[n]} (f(\mathbf{x}_v)-f(0))$, which agrees with the classic definition $\sum_{v\in[n]} f(\mathbf{x}_v)$ when $f(0)=0$. }  
Step ($\star\star$) allows us to extract useful information from $f(X)$ about $f(\mathbf{x})$.

Now suppose that $f$ is some function that can be decomposed into a linear combination of homomorphisms $f_1,f_2,f_3,\ldots$ from $(\Z,+)$ to $(\C,\times)$ after calibration, i.e.,  $f(z)=\sum_{j}(f_j(z)-1)$.
We can estimate each $f_j$-moment separately for each $j$ and thereby obtain an estimate of the $f$-moment by combining the estimates of the components.
See Figure \ref{fig:diagram} for a visualization of the estimation process.

\begin{figure}[!ht]
    \centering
    \includegraphics[width=0.8\textwidth]{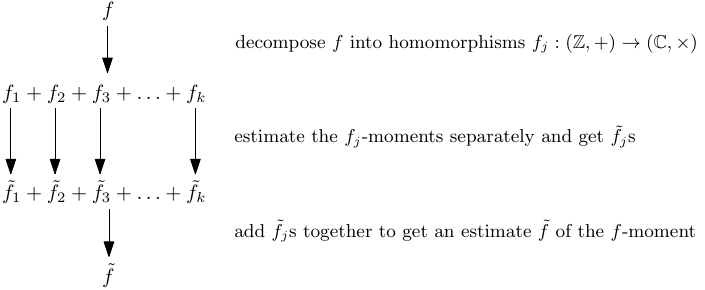}
    \caption{Diagram of the $f$-moment estimation process. Decompose $f$ as $f_1+\ldots+f_k$ where $f_j$s are homomorphisms. For $j=1,\ldots,k$, compute $\tilde{f}_j$ that estimates the $f_j$-moment of $\mathbf{x}$.  Then $\tilde{f}=\tilde{f}_1+\ldots+\tilde{f}_k$ is the estimator for the $f$-moment.
    }
    \label{fig:diagram}
\end{figure}

There is a famous tool to decompose functions into a linear combination of homomorphisms, namely the \emph{Fourier transform}. Nevertheless, the Fourier transform requires $f$ to be integrable which does not cover common functions like $f(x)=|x|^p$, $p\in(0,2)$. A few more tricks are needed to construct the final sketch though, this decomposition-into-homomorphisms trick is the key to overcome the barriers that held 
back previous sampling based methods. 

\subsection{Organization}
In \cref{sec:technicalintro}, we give a technical overview with our formal results in \cref{sec:new-results}. 
In \cref{sec:choose_t}, we introduce two additional tricks in the construction and estimation of the main sketch \SymmetricPoissonTower{}.
In \cref{sec:harmonic_ana}, we formally define the \SymmetricPoissonTower{} sketch with a full analysis of the estimation for individual harmonic moments. We analyze the combined estimators in \cref{sec:ana_comb}. We assemble the proof of the main theorem in \cref{sec:main_proof}.

\section{Technical Introduction \& Related Work}\label{sec:technicalintro}

\subsection{(Symmetric) Harmonic Moments}\label{sec:harmonic_intro}
We now formalize the estimation procedure. Note that the set of bounded homomorphisms from $(\Z,+)$ to $(\C,\times)$ are nothing secret but just $\{g_\gamma(z)=e^{i\gamma z}:\gamma\in [0,2\pi)\}$.\footnote{Technically, $g_r\colon z\mapsto e^{rz}$ is also a multiplicative homomorphism parameterized by $r\in \R$. However, we will not use $g_r$ since we only consider symmetric functions in this work. Nevertheless, $g_r$ might be useful when considering incremental-only streams. } The scheme in \cref{fig:diagram} decomposes $f(z)$ to $\int_0^{2\pi}(e^{i\gamma z }-1)\,\mu(d\gamma)$ where $\mu$ is some signed measure.  By convention \cite{braverman2010zero,braverman2016streaming}, we will focus on functions that are symmetric and non-negative. Thus for now, we will only consider functions $f$ that can be decomposed to  $\int_0^{2\pi}(1-\cos(\gamma z ))\,\nu(d\gamma)$ where $\nu$ is a \emph{positive} measure.
\begin{definition}[(symmetric) harmonic moments]\label{def:harmonic_moment}
For any $\gamma>0$, define $f_\gamma(x)=1-\cos(\gamma x)$. The \emph{$\gamma$-harmonic moment} of a stream with frequency vector $\mathbf{x}\in \Z^n$ is
\begin{align*}
    f_\gamma(\mathbf{x}) &=\sum_{v\in[n]} f_\gamma(\mathbf{x}_v) = \sum_{v\in[n]}(1-\cos(\gamma \mathbf{x}_v)).
\end{align*}
In other words, the $\gamma$-harmonic moment is the $f_\gamma$-moment.
\end{definition}

The computation around ($\star\star$) shows that the $e^{i\gamma z}$-moment can be estimated by Poissonized subsampling. Since $\cos(\gamma z)=(e^{i\gamma z}+e^{-i\gamma z})/2$, the symmetric harmonic moments can be estimated similarly with  \emph{symmetric Poisson random variables}.

\begin{definition}[symmetric Poisson distribution]\label{def:sym_poisson}
A symmetric Poisson random variable with rate $\lambda$ distributes as the difference of two independent Poisson random variables with rate $\lambda/2$. This distribution is denoted $\SymmetricPoisson(\lambda)$. 
\end{definition}
Consider the following 1-cell linear sketch, parameterized by a \emph{subsampling-level} $p>0$.
\begin{align*}
    R_p &= \sum_{v\in[n]}\mathbf{x}_v Y_v,
\end{align*}
where $Y_v \sim \SymmetricPoisson(p)$ are i.i.d.~random variables, provided by a random oracle.\footnote{Throughout the paper we work in the \emph{random oracle} model, 
in which we can evaluate uniformly random hash functions 
$H : [n]\to [0,1]$ and $Y_v$ can be simulated through the randomness in $H(v)$. The hash function is \emph{not} stored by the sketch but provided externally.} 
The following lemma symmetrizes the computation around ($\star\star$).

\begin{lemma}\label{lem:observation}
For any $\gamma >0$, 
\begin{align*}
    \E e^{i\gamma R_p}=e^{-pf_\gamma(\mathbf{x})}.
\end{align*}
\end{lemma}
\begin{proof}
By the computation around ($\star\star$), for $Y\sim \SymmetricPoisson(p)$ and $Z,Z'\sim \Poisson(p/2)$, we have
\begin{align}
    \E e^{i\gamma Y} &= \E e^{i\gamma (Z-Z')} = \E e^{i\gamma Z} \E e^{-i\gamma Z'}= e^{-\frac{p}{2}(1-e^{i\gamma})}e^{-\frac{p}{2}(1-e^{-i\gamma})}= e^{-p(1-\cos(\gamma))}. \label{eq:symmetric_poisson}
\end{align}
We can now relate $R_p$ to the $f_\gamma$-moment as follows.
\begin{align*}
    \E e^{i\gamma R_p} &= \E e^{i\gamma \sum_{v\in[n]}\mathbf{x}_v Y_v}\\
    &=\prod_{v\in[n]} \E e^{i\gamma \mathbf{x}_v Y_v} \tag{$Y_v$s are independent}\\
    &=\prod_{v\in[n]} e^{-t(1-\cos(\gamma \mathbf{x}_v))}\tag{$Y_v\sim \SymmetricPoisson(t)$ and \cref{eq:symmetric_poisson}}\\
    &= e^{-t\sum_{v\in[n]}(1-\cos(\gamma \mathbf{x}_v))}=e^{-pf_\gamma(\mathbf{x})}. \tag{\cref{def:harmonic_moment}}
\end{align*}
\end{proof}
With the identity in \cref{lem:observation}, one may estimate
the $\gamma$-harmonic moment $f_\gamma(\mathbf{x})$ by choosing a suitable level $p=\Theta(1/f_\gamma(\mathbf{x}))$. As the ``angular rate'' $\gamma$ and the input $\mathbf{x}$ vary, the size of $f_\gamma(\mathbf{x})$ can vary and therefore we need to maintain $R_p$ for $O(\log n)$ geometrically spaced levels $p$. We thus can $(1\pm\epsilon)$-approximate the $f_\gamma$-moments for all $\gamma>0$ with $O(\epsilon^{-2}\log n)$ cells.\footnote{Since the input $\mathbf{x}$ is integer-valued, $f_\gamma (\mathbf{x}) = f_{\gamma + 2k\pi}(\mathbf{x})$ for any $k\in \Z$. We assume $f_\gamma(\mathbf{x})=\Omega(1)$ here for simplicity.}  Further details are in \cref{sec:harmonic_ana}. 

To show how this approach already significantly expands the universality, we prove that, surprisingly, for \emph{almost all} $f_\gamma$-moments, sampling based methods \cite{braverman2015universal,braverman2016streaming} need polynomial space (exponentially worse than the harmonic approach above!)  because $1-\cos(\gamma x)$ can occasionally drop to $O(1/\poly(x))$.
\begin{lemma}\label{lem:why_fail}
Let $\gamma\in\R_+$ such that $\gamma/\pi\not\in \mathbb{Q}$ and
let $f_\gamma(x)=1-\cos(\gamma x)$. Then
\begin{align*}
    \forall k\in\Z_+, \exists w_k \in\{1,\ldots,k\}, f_\gamma(w_k) = O(1/k^2) \tag{$f_\gamma$ drops polynomially}
\end{align*}
\end{lemma}
\begin{remark}
    Recall that a function $f$ drops polynomially if $0<\min_{z\in [M]}f(z)\leq O(1/\poly(M))$ as $M\to\infty$.
\end{remark}
\begin{proof}
We use the following classic pigeonhole argument. Let $\beta = \gamma/(2\pi)$.
  For any $k\geq 1$, by the pigeonhole principle, there exist distinct $a,b\in \{1,2,3,\ldots,k+1\}$, such that $|(a\beta-\floor{a\beta})-(b\beta -\floor{b\beta})|\leq 1/k$. Therefore, 
  \begin{align*}
      f_{\gamma}(a-b)&=1-\cos((a-b)\gamma)=1-\cos((a\beta-b\beta)2\pi)=1-\cos(((a\beta-\floor{a\beta})-(b\beta -\floor{b\beta}))2\pi)\\
      &\leq 2\pi^2((a\beta-\floor{a\beta})-(b\beta -\floor{b\beta}))^2\leq 2\pi^2/k^2,
  \end{align*}
  where we used the fact that $1-\cos(z)\leq z^2/2$ for any $z\in \R$. Define $w_k=|a-b|\in\{1,\ldots,k\}$ and we have $f_\gamma(w_k)\leq 2\pi^2/k^2$.
\end{proof}

In addition, the harmonic approach immediately gives an \emph{explanation} why the function $g_{np}(x)$  in \cite[\S 5]{braverman2016streaming} (defined in \cref{eq:gnp}) occasionally drops to $O(1/x)$ but is still tractable by some ad hoc sketch: the $g_{np}$-moment can be decomposed into harmonic moments. 
\begin{align*}
    g_{np}(x) &= \frac{4}{3}\sum_{\gamma \in \mathbb{B}}(1-\cos(2\pi \gamma x))2^{-2\tau_*(\gamma)}\tag{the nearly periodic function example}
\end{align*}
The set $\mathbb{B}$ is the set of real numbers in $[0,1)$ with finite binary representation. 
The function $\tau_*(\gamma)=\min\{j\in\N\mid 2^{-j}|\gamma\}$ returns the length of the binary representation of $\gamma$. See \cref{sec:decomp_gnp} for a justification of the decomposition of $g_{np}$. Given the decomposition, $g_{np}$ can be automatically estimated by combining the estimates of harmonic moments and this approach is much more space-efficient than the ad hoc sketch in \cite{braverman2016streaming}. We will discuss why the estimates can be safely combined without blowing up the error in \cref{sec:combine}.

\medskip

So far, it may seem that the harmonic approach is only useful in the theoretical pursuit of fully characterizing tractability. However, it turns out that the information scattered in the harmonic moments can be combined to answer meaningful queries ($L_p$, logarithm, etc.), making this approach an improved \emph{practical} sketch \emph{without any overhead for sampling}. We now list the decompositions for  common functions. See \cref{fig:harmonic_decomposition} for visualization.
\begin{align*}
    |x|^p &= \int_0^\infty (1-\cos(\gamma x ))\frac{1}{(-\Gamma(-p))\cos(p\pi/2)\gamma^{1+p}} \,d\gamma \tag{$L_p$-moment, $p\in(0,1)\cup(1,2)$}\\ 
    |x| &= \int_0^\infty (1-\cos(\gamma x ))\frac{2}{\pi\gamma^{2}} \,d\gamma \tag{$L_1$-moment}\\
    \ind{x\neq 0} &= \int_0^\pi (1-\cos(\gamma x ))\frac{1}{\pi} \,d\gamma \tag{$L_0$-moment/support size, $x\in\Z$}\\
    1-e^{-r |x|} &= \int_0^\infty (1-\cos(\gamma x ) )\frac{2 r}{\pi(\gamma^2+r^2)}\,d\gamma \tag{symmetric ``soft cap'' moment}
\end{align*}

As seen above, all $L_p$-moments ($p\in[0,2)$) are within the span of the harmonic moments. The symmetric version of Cohen's basis functions \cite{cohen2017hyperloglog,cohen2018stream} $G_r^*(x)=1-e^{-r|x|}$ are in the span of the harmonic moments, as well as their linear combinations.  For example,
\begin{align}    
\log(1+|x|)&=\int_0^\infty (1-e^{-r|x|})e^{-r}r^{-1}\,dr\nonumber\\
&= \int_0^\infty \left[ \int_0^\infty (1-\cos(\gamma x) )\frac{2 r}{\pi(\gamma^2+r^2)}\,d\gamma\right] e^{-r}r^{-1}\,dr\nonumber\\
&= \int_0^\infty (1-\cos(\gamma x))  \left[ \int_0^\infty \frac{2 e^{-r}}{\pi(\gamma^2+r^2)}\,dr\right] \,d\gamma\tag{logarithmic-moment}\nonumber 
\end{align}

The $L_2$-moment, on the other hand, lies on the ``border'' of the harmonic moments, since $|x|^2 = \lim_{\gamma\to 0} \frac{2}{\gamma^2}(1-\cos(\gamma x))$.  \cref{sec:L2} explains how to estimate the $L_2$-moment with the \SymmetricPoissonTower.

\begin{figure}[!p]
    \centering
    \begin{subfigure}{0.4\linewidth}
    \includegraphics[width=\textwidth]{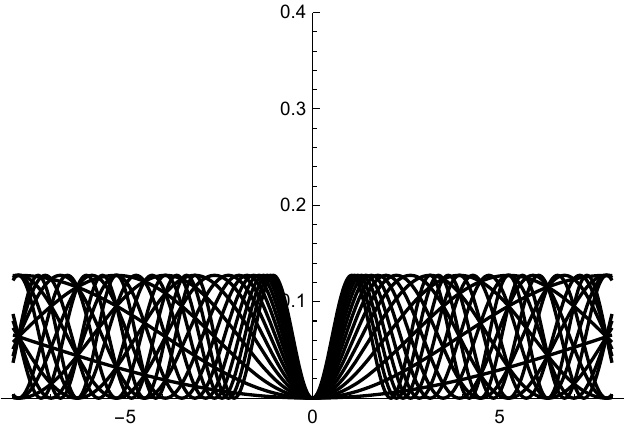}
    \caption{Harmonic decomposition of $L_0$}
    \end{subfigure}
    \begin{subfigure}{0.4\linewidth}
    \includegraphics[width=\textwidth]{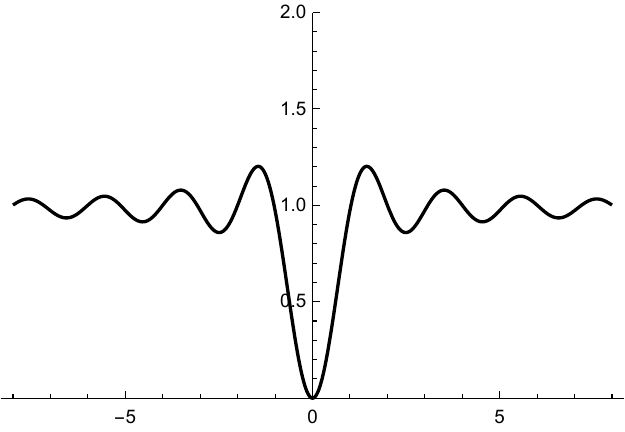}
    \caption{Sum back $f(z)=\ind{z\neq 0}$, $z\in\Z$}
    \end{subfigure}
    \begin{subfigure}{0.4\linewidth}
    \includegraphics[width=\textwidth]{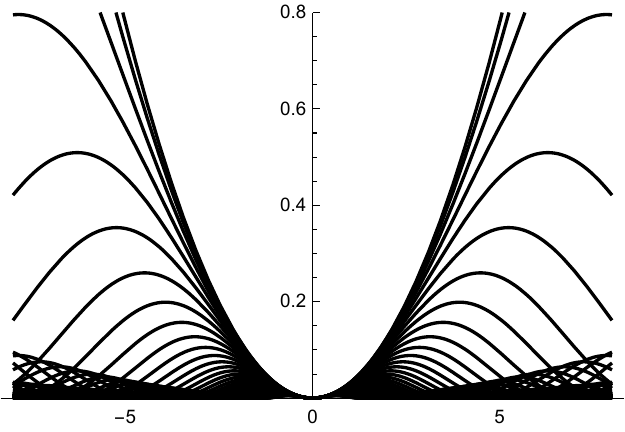}
    \caption{Harmonic decomposition of $L_1$}
    \end{subfigure}
    \begin{subfigure}{0.4\linewidth}
    \includegraphics[width=\textwidth]{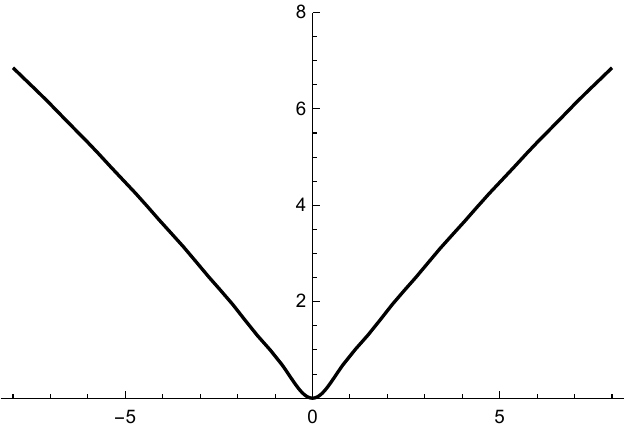}
    \caption{Sum back to $f(z)=|z|$, $z\in\Z$}
    \end{subfigure}
    \begin{subfigure}{0.4\linewidth}
    \includegraphics[width=\textwidth]{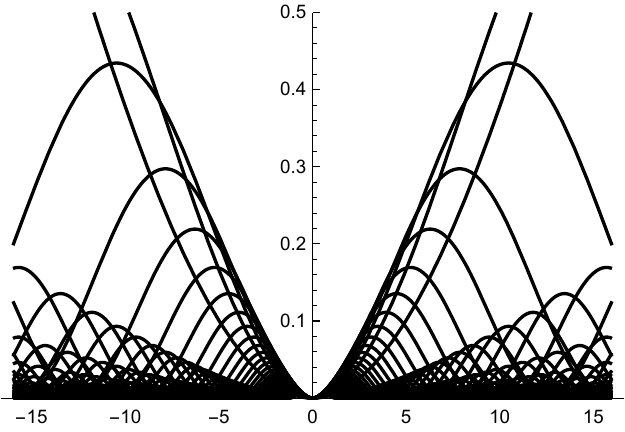}
    \caption{Harmonic decomposition of logarithm}
    \end{subfigure}
    \begin{subfigure}{0.4\linewidth}
    \includegraphics[width=\textwidth]{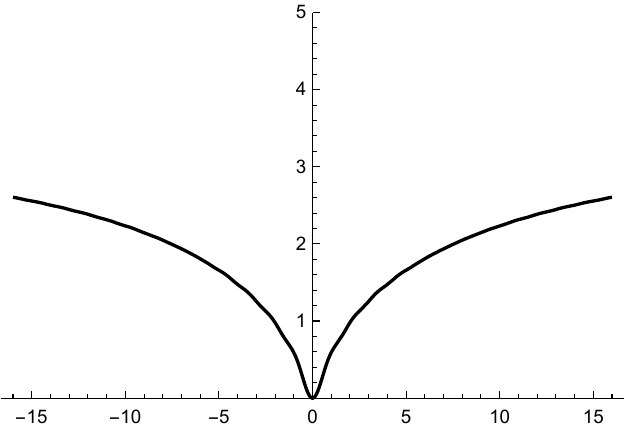}
    \caption{Sum back to $f(z)=\log(1+|z|)$, $z\in\Z$}
    \end{subfigure}
    \begin{subfigure}{0.4\linewidth}
    \includegraphics[width=\textwidth]{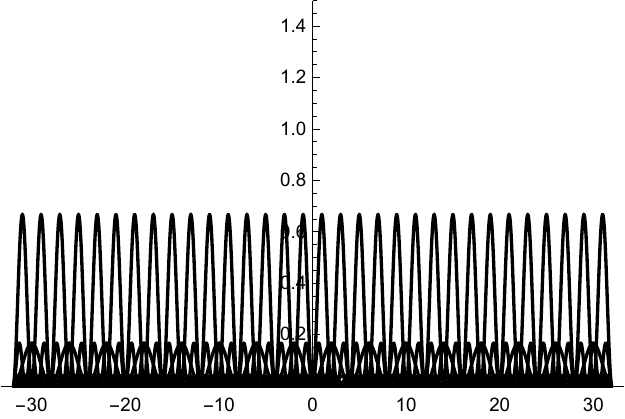}
    \caption{Harmonic decomposition of $g_{np}$}
    \end{subfigure}
    \begin{subfigure}{0.4\linewidth}
    \includegraphics[width=\textwidth]{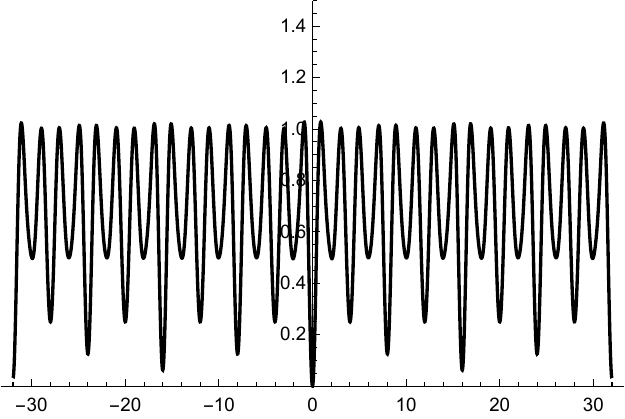}
    \caption{Sum back to $g_{np}(z)=2^{-\tau(|z|)}$, $z\in\Z$}
    \end{subfigure}
    \caption{Diagram of harmonic decomposition: Different function moments can be estimated by combining estimates of harmonic moments with different weights. Only a subset of harmonic components are used for visualization and they only sum back to an approximation of $f$. As more harmonic components are used, the sum will uniformly converge to $f$ over $[-M,M]$. }
    \label{fig:harmonic_decomposition}
\end{figure}

\subsection{Combination of Harmonic Moments}\label{sec:combine}
We now discuss the error of the combined estimator.
Let $\{V_\gamma\}_{\gamma\in \R_+}$ be a family of estimates where $V_\gamma$ is an approximation of the $f_\gamma$-moment. Let $f(x)=\int_0^\infty f_\gamma(x)\,\nu(d\gamma)$ where $\nu$ is some positive measure. We define the \emph{combined estimator} of $f$ as $\int_0^\infty V_\gamma \,\nu(d\gamma)$.

At first sight, it might look like the combined estimator can have a large error, since the estimates $\{V_\gamma\}_{\gamma\in \R_+}$ are \emph{arbitrarily correlated}. They are produced with a \emph{single} universal sketch! However, Cohen \cite{cohen2017hyperloglog,cohen2018stream} observed that as long as the \emph{mean} and \emph{variance} of each individual estimate $V_\gamma$ have good guarantees, such guarantees can be transported to the combined estimator. In \cite{cohen2017hyperloglog,cohen2018stream}, Cohen developed a universal sketch over \emph{incremental streams} with a min-based sampling scheme.\footnote{Min-based schemes are not applicable over turnstile streams.} In contrast to the turnstile model, incremental streams have a frequency vector $\mathbf{x}$ in $\N^n$ and only positive updates are allowed. For any $r>0$, let $G_r(x)=1-e^{-rx}$. 
Cohen's sketch has a collection of estimators $\{V_r^+\}_{r>0}$ 
such that for any $r>0$,
\begin{align}
    \E V_r^+ &= G_r(\mathbf{x}) = \sum_{v\in[n]} (1-e^{-r\mathbf{x}_v}) \label{eq:cohen_mean}\\
    \var V_r^+ &= O(\epsilon^{2}) G_r(\mathbf{x})^2. \label{eq:cohen_var}
\end{align}
For any positive measure $\nu$ on $(0,\infty)$, let $G_{\nu}(x)=\int_0^\infty G_r(x) \,\nu(dr)$.\footnote{For $G_{\nu}$ to exist, the measure $\nu$ has to satisfy that $\int_0^\infty \min(r,1)\,\nu(dr)<\infty$.} Cohen noted \cite{cohen2017hyperloglog,cohen2018stream} that
\begin{align*}
\E \int_0^\infty V_r^+ \,\nu(dr) &=  \int_0^\infty \E  V_r^+\,\nu(dr) = \sum_{v\in[n]} G_\nu(\mathbf{x}_v) = G_\nu(\mathbf{x}) \tag{\cref{eq:cohen_mean}
}\\
\var \int_0^\infty V_r^+ \,\nu(dr) &=  \int_0^\infty \int_0^\infty \Cov(V_r^+,V_{r'}^+)\,\nu(dr)\nu(dr')\\
&\leq \int_0^\infty \int_0^\infty \sqrt{\var V_r^+}\sqrt{\var V_{r'}^+}\,\nu(dr)\nu(dr') \tag{Cauchy-Schwartz} \\
&= \left(\int_0^\infty \sqrt{ \var  V_r^+} \,\nu(dr)\right)^2 \\
&= \left(\int_0^\infty O(\epsilon) G_r(\mathbf{x}) \,\nu(dr)\right)^2 \tag{\cref{eq:cohen_var}}\\
&= O(\epsilon^{2}) G_\nu(\mathbf{x})^2.
\end{align*}
Therefore, by Chebyshev's inequality, $\int_0^\infty V_r^+ \,\nu(dr) $ is a $(1\pm \epsilon)$-approximation of the $G_\nu$-moment $G_\nu(\mathbf{x})$.  In this work we will construct estimators whose mean and variance can be analyzed so that we can apply Cohen's argument to our harmonic framework.

\subsection{Formal Statement of New Results}\label{sec:new-results}

We now formally state the main theorem. We first define the function family 
\[
\mathcal{F}_*=\{f(x)=cx^2+\int_0^\infty (1-\cos(\gamma x))\,\nu(d\gamma)\mid c\in\R_+, \text{$\nu$ is a positive measure}\},
\]
which is the \emph{mathematical} span of the harmonic moments as well as $x^2$. The simple take-away message is that 
the universal harmonic sketch is \emph{almost} $\mathcal{F}_*$-universal: It can $(1\pm \epsilon)$-approximate any $f$-moment with $f\in \mathcal{F}_*$, conditioning on $f(\mathbf{x})\in [1/\poly(n), \poly(n)]$. Nevertheless, to be fully rigorous, we need a few more technical conditions on $f$ that guarantee  $f(\mathbf{x})\in [1/\poly(n), \poly(n)]$ for any $\mathbf{x}\in [-M,M]^n$.

For any $f\in \mathcal{F}_*$, define $C_f = \int_0^\infty \min(\gamma^2,1)\,\nu(d\gamma)$. Note that  for any $x\in\R$, $f(x)$ is finite if and only if $C_f<\infty$. Recall $M\in \Z_+$ is the frequency bound. Define $\xi(f,M)=\min_{x\in[M]}\{f(x)\mid f(x)> 0\}$.

\begin{theorem}[Universal Harmonic Sketch]\label{thm:main}
There exists an $O(\epsilon^{-2}\log^2 (n\epsilon^{-1}M))$-bit space sketch (\SymmetricPoissonTower) 
such that for any turnstile stream with frequency vector $\mathbf{x}\in \{-M,\ldots,M\}^n$, the following statements are true.
\begin{description}
    \item[Combination] For any function $f\in\mathcal{F}_*$ such that $C_f/\xi(f,M)=O(\poly(nM\epsilon^{-1}))$, a $(1\pm \epsilon)$-approximation of the $f$-moment can be returned with probability 2/3.
    \item[Signed combination] For any functions $f,g\in\mathcal{F}_*$ such that $f(x)-g(x)=\Omega(f(x)+g(x))$ for $x\in [M]$ and $(C_f+C_g)/\xi(f-g,M)=O(\poly(nM\epsilon^{-1}))$, a $(1\pm \epsilon)$-approximation of the $(f-g)$-moment can be returned with probability 2/3.  
\end{description}
\end{theorem}
\begin{remark}
    In the usual region where $\epsilon^{-1} =O(\poly(n))$ and $M=O(\poly(n))$, the space is $O(\epsilon^{-2}\log^2n)$. See \cref{fig:summary} for a summary.
\end{remark}
\begin{remark}\label{rem:boundary}
To further extend the universality, one may set $\tilde{h}(M)=(\max_{x\in [M]} \frac{f(x)+g(x)}{f(x)-g(x)})^2$. For a non-negative function $f^*=f-g$, the $f^*$-moment can be $(1\pm \epsilon)$-approximated by the signed combination with $O(\tilde{h}(M)\cdot\epsilon^{-2} \log^2n)$ bits of space.
\end{remark}

We now discuss the parameters in the conditions of the main theorem.
\begin{description}
    \item[Function constant $C_f$.] For any \emph{fixed} function $f\in \mathcal{F}_*$, $C_f$ is a constant. For example, one may compute that for the $L_p$-moments where $f(x)=|x|^p$, the constants $C_{|x|^p}\leq 2$ for any $p\in(0,2)$ and for the $L_0$-moment, $C_{\ind{x\neq 0}}<1$. For the $g_{np}$ function,  $C_{g_{np}}\leq \frac{2}{3}$; see \cref{sec:decomp_gnp}. 
    However, given that the sketch is universal, 
    when a \emph{family} of functions $f$ are considered, the constraint $C_f=O(\poly(nM\epsilon^{-1}))$ may matter.
    \item[Drop rate $\xi(f,M)$.] For any \emph{increasing} function $f\in \mathcal{F}_*$, $\xi(f,M)=f(1)$ is a constant. The constraint $\xi(f,M)=\Omega(1/\poly(nM\epsilon^{-1}))$  only matters when $f$ can become very small over the frequency range $\{1,\ldots,M\}$. Note that this condition in the main theorem is much weaker than the corresponding drop rate condition in previous techniques \cite{braverman2010zero,braverman2015universal,braverman2016streaming} which will fail whenever $\xi(f,M)$ is not polylog. In our main theorem, $\xi(f,M)$ can decrease with any polynomial rate. This allows the harmonic sketch to overcome the previous barrier to approximate $g_{np}$ for which $\xi(g_{np},M)=\Theta(1/M)$ and the golden harmonic function $g_{gold}$ (\cref{lem:gold}, below) for which $\xi(g_{gold},M)=\Theta(1/M^2)$. 
\end{description}

Does the harmonic approach completely dominate the $L_2$-heavy hitter method \cite{braverman2016streaming} in terms of universality? This is a difficult statement to check as the criteria in \cite{braverman2016streaming} is a set of properties (slow-jumping, slow-dropping, and predictable) coming from a deep analysis of the Indyk-Woodruff subsample-and-CountSketch technique \cite{IndykW03}.  On the other hand, the new harmonic approach sets criteria on the harmonic structure of the functions. However, we do know the following.
\begin{itemize}
    \item Common target functions like $L_p$-moments ($f(x)=|x|^p$, $p\in[0,2]$), logarithmic-moment ($f(x)=\log(1+|x|)$), and ``soft cap''-moment ($f(x)=1-e^{-|x|}$) all have simple harmonic structures and thus can be tracked efficiently with the \SymmetricPoissonTower{} sketch. Thus the \SymmetricPoissonTower{} is a provably more space efficient universal sketch for the common tasks above, removing all sampling-related data structures.
    \item There exist functions that can be tracked with polylog space using the harmonic approach, while sampling based approaches need polynomial space. One example is the $g_{np}$ function from \cite{braverman2016streaming},
which has an \emph{ad hoc} solution, but is covered by 
our main \cref{thm:main} using a much smaller space. For a new example, we also prove that the \emph{golden harmonic moment} is 
not slow-dropping~\cite{braverman2016streaming},
 but covered by \cref{thm:main}. 
See \cref{lem:gold} below (the proof is in \cref{sec:gold}).
\end{itemize}

In all, a provably ``truly universal'' sketch has not been found yet, but the harmonic approach is a significant step towards the final goal.

\begin{lemma}[golden harmonic moment]\label{lem:gold}
Let $g_{gold}(x) = 1-\cos(\frac{1+\sqrt{5}}{2}\cdot 2\pi x)$ for $x\in \Z$. Then
\begin{itemize}
    \item $g_{gold}$ is not slow-dropping since $g_{gold}(1)>1.7$ and $\xi(g_{gold},M)=O(1/M^2)$. 
    \item $g_{gold}$ is covered by \cref{thm:main} 
    since $C_{g_{gold}}=1$ and $\xi(g_{gold},M) =\Omega(1/M^2)$. 
\end{itemize}
\end{lemma}

\begin{figure}[!p]
    \centering    \includegraphics{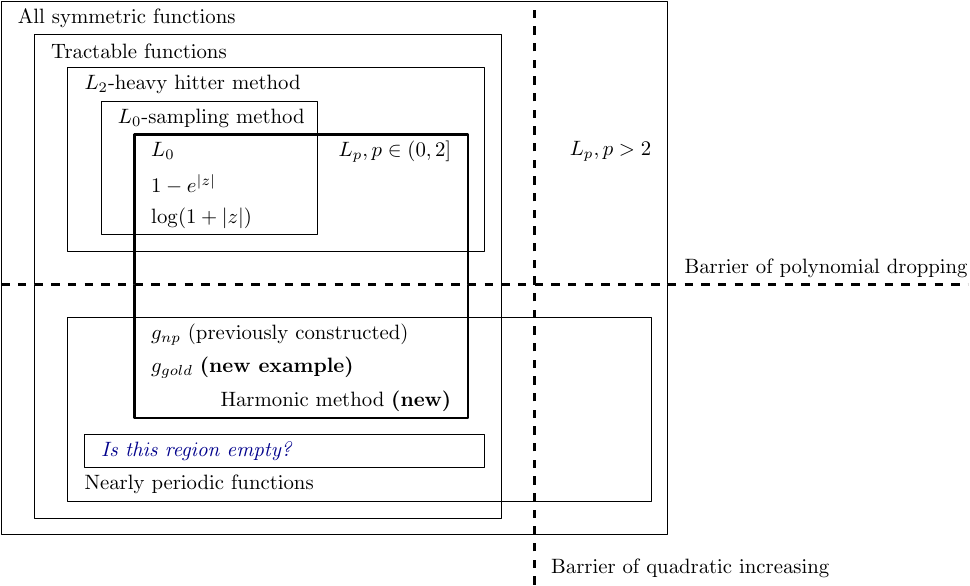}\\
    \vspace{0.2in}
    \begin{tabular}{|c|c|}
    \hline
  sketch  & space  \\
\hline
$L_0$-sampling based \cite{braverman2015universal}  &  $O(h^*(M)\cdot\epsilon^{-2}\log^2 n)$  \\
$L_2$-heavy hitter based \cite{braverman2016streaming} & $O(h(M)\cdot \epsilon^{-2}\log^6 n\log\log n)$ \\
ad hoc sketch for $g_{np}$ \cite{braverman2016streaming} & $O(\epsilon^{-8}\log^{15} n)$ \\
harmonic approach \textbf{(new)} &  $O(\tilde{h}(M)\cdot\epsilon^{-2}\log^2 n)$ \\
\hline
\end{tabular}
    \caption{Summary of tractability and universality. Pivotal functions are shown to demonstrate the estimation power of different universal sketching schemes.  See \cite{braverman2010zero,braverman2016streaming} for $L_2$-heavy hitter based methods ($h$ satisfies $f(y)/f(x)\in[1/h(y),(y/x)^2 h(y)]$ for any $0<x<y\leq M$) and \cite{braverman2015universal,chestnut2015stream} for $L_0$-sampling based methods ($h^*(n)\geq \max_{j\in[M]}f(j)/\min_{j\in [M]}f(j)$). The functions $h$ and $h^*$ describe the growing overhead to approximate harder and harder $f$-moments (``hard'' with respect to the corresponding scheme). The harmonic approach has a different ``boundary behavior'' with overhead $\tilde{h}$ which is described in \cref{rem:boundary}.   The $g_{np}$ (defined in \cref{eq:gnp}) is constructed in \cite{braverman2016streaming} to show the existence of tractable functions with occasional polynomial dropping ($g_{np}(z)=1/z$ if $z$ is a power of $2$). $g_{gold}$ is newly constructed in \cref{lem:gold} which occasionally drops quadratically ($g_{gold}(z)=O(1/z^2)$ for infintely many $z$).   It is proved in \cite{bar2004information} that $L_p$ is not tractable if $p>2$.  }
    \label{fig:summary}
\end{figure}

\subsection{Related Work}
B{\l}asiok, Braverman, Chestnut, Krauthgamer, and Yang \cite{blasiok2017streaming} further extend the Indyk-Woodruff \cite{IndykW03} heavy hitter technique to estimate \emph{symmetric norms} of the dynamic vector $\mathbf{x}$. Such heavy hitter based methods are proved to be optimal for all symmetric norms up to polylog factors. However, they face the same barrier as \cite{braverman2016streaming} when trying to estimate generic $f$-moments that do not correspond to norms. 

\medskip

The idea of ``solving'' or ``inverting'' hash collisions is not entirely new. Especially in the incremental streaming model, Kumar, Sung, Xu, and Wang \cite{kumar2004data} have used the \emph{Expectation Maximization} technique to iteratively estimate $\phi_1,\phi_2,\ldots$ where $\phi_j$ is the number of elements with value $j$. A more recent work by Chen, Indyk, and Woodruff \cite{chen2024space} estimates $\phi_1,\phi_2,\ldots$ iteratively with dynamic programming. However, both \cite{kumar2004data} and \cite{chen2024space} crucially  rely  on the incremental structure, so, e.g., if a ``0'' is observed, it must be an empty cell; if a ``1'' is observed, it has to be a singleton with value 1; if a ``2'' is observed, it is either the sum of two ones or a singleton with value two; and so on. It is not clear whether such techniques \cite{kumar2004data,chen2024space} can be generalized to the turnstile model since if elements have both positive and negative values, then, e.g., a ``0'' cell can be either empty, $(1) + (-1)$, or $(2) + (-3) + (1)$, etc.---there are infinitely many combinations of integer-values to collide into a zero. From this perspective, our insight in \cref{sec:insight} is really powerful and potentially provides a new simple scheme for incremental streams as well (with multiplicative homomorphism $f_r^*\colon z\mapsto e^{-rz}$ which is bounded now since the argument $z$ is non-negative in the incremental model). 

\medskip

As discussed in \cref{sec:combine}, the idea of combining estimates of ``basis moments'' to estimate functions that are in the ``span'' has been first used by Cohen \cite{cohen2017hyperloglog,cohen2018stream} in the incremental setting. Nevertheless, the context is very different---Cohen uses min-based, multi-objective samplers to estimate the basis moments, whereas we consider the more generic turnstile streams (with both insertion and deletion)  and our basis are the harmonic moments.

\medskip

Though we focus on the polylog-space regime in this work, there have been efforts on proving optimal space in the polynomial regime as well. In particular, upper bounds and lower bounds for large frequency moments $(L_p,p>2)$ are intensively studied in \cite{AlonMS99,braverman2014optimal,indyk2005optimal,chakrabarti2003near,price2012applications,ganguly2004estimating,coppersmith2004improved,braverman2013approximating}.

\section{Implicit Level Selection and Smoothed Subsampling}\label{sec:choose_t}

The take-away message from~\cref{sec:harmonic_intro} is that if one 
maintains the linear sketch  $R_p = \sum_{v\in[n]}\mathbf{x}_v Y_v$,
where $Y_v\sim \SymmetricPoisson(p)$, 
then for any $\gamma >0$, 
\begin{align*}
\E e^{i\gamma R_p} =e^{-pf_\gamma(\mathbf{x})}.\tag{\cref{lem:observation}}
\end{align*}
By selecting $p=\Theta(1/f_\gamma(\mathbf{x}))$, the $\gamma$-harmonic moment $f_\gamma(\mathbf{x})$ can be $(1\pm \epsilon)$-approximated with $O(\epsilon^{-2})$ i.i.d.~copies of $R_p$. 

To minimize the space complexity, we will select the suitable level $p$ \emph{implicitly}.\footnote{Kane, Nelson, and Woodruff faced a similar \emph{level selection} problem in estimating distinct elements~\cite{KaneNW10}. 
Their approach would be to choose $t$ from a \emph{rough estimator} of $f_\gamma(\mathbf{x})$, which would return a large constant factor
approximate in the range $[c_1f_\gamma(\mathbf{x}), c_2f_\gamma(\mathbf{x})]$ with probability $1-\delta$. 
However, here the \emph{variance} is complicated to bound with such explicit level selectors. 
A rough bound for the increase in variance for failed selections is $O(n^2)$ (note that $f_\gamma(\mathbf{x})=O(n)$), which  will require 
$\delta = O(1/\poly(n))$ to suppress. This will  add an additional $\log (n)$ factor to the algorithm.} 
Instead of using auxiliary data structures to select a proper level $p$, we maintain levels\footnote{While it is more common in the streaming literature \cite{flajolet1985probabilistic} to maintain $m$ i.i.d.~cells at levels $2^{-k}$, $k\in\Z$, here we sample a \emph{single} cell at \emph{finer} levels $e^{-k/m}$, $k\in\Z$. This trick is known as \emph{smoothing} in \cite{pettie2021information}, which is used for removing the oscillating component in the estimators.} $R_p$ with $p=e^{-k/m}$, $k\in\Z$, and aggregate them as follows.
\begin{align*}
    U_\gamma &= \sum_{k\in \Z} (1-e^{i\gamma R_{e^{-k/m}}}) e^{\frac{1}{3}k/m},
\end{align*}
whose mean is\footnote{Later on we will establish that $U_\gamma$ is absolutely convergent almost surely and the summation and expectation can be interchanged.} 
\begin{align*}
    \E U_\gamma &= \sum_{k\in \Z} \E (1-e^{i\gamma R_{e^{-k/m}}})  e^{\frac{1}{3}k/m}\\
    &=\sum_{k\in \Z}(1-e^{-e^{-k/m}f_\gamma(\mathbf{x})})  e^{\frac{1}{3}k/m} \tag{\cref{lem:observation}}\\
    &=f_\gamma(\mathbf{x})^{1/3}\sum_{k\in \Z}(1-e^{-e^{-(k-m\log f_\gamma(\mathbf{x}))/m}})  e^{\frac{1}{3}(k-m\log f_\gamma(\mathbf{x}))/m}\\
    &\approx f_\gamma(\mathbf{x})^{1/3}\sum_{k\in \Z}(1-e^{-e^{-k/m}})e^{\frac{1}{3}k/m} \tag{Shift $k\gets \floor{k-m\log f_\gamma(\mathbf{x})}$}
\end{align*}
Roughly speaking we have $\E U_\gamma\propto f_\gamma(\mathbf{x})^{1/3}$. Let $U_\gamma^{(1)},U_\gamma^{(2)},U_\gamma^{(3)}$ be three i.i.d.~copies of $U_\gamma$. We know $U_\gamma^{(1)}U_\gamma^{(2)}U_\gamma^{(3)}$ 
will be a roughly unbiased estimator for $f_\gamma(\mathbf{x})$ 
after normalization. 
The reason that this sum implicitly chooses the correct levels is because the contribution to the mean decays exponentially outside 
the suitable levels. 
\begin{itemize}
    \item When $k\gg \log_2 f_\gamma(\mathbf{x})$, $\left(1-e^{-2^{-k}f_\gamma(\mathbf{x})}\right) 2^{k/3}\approx 2^{-k}f_\gamma(\mathbf{x}) 2^{k/3}=f_\gamma(\mathbf{x}) 2^{-2k/3}$. Thus the contribution to the mean $\E U_\gamma$ is vanishing exponentially as $k\to\infty$.
    \item When $k\ll  \log_2 f_\gamma(\mathbf{x})$, $\left(1-e^{-2^{-k}f_\gamma(\mathbf{x})}\right) 2^{k/3}\approx  2^{k/3}$. Thus the contribution to the mean is also vanishing exponentially as $k\to -\infty$.
\end{itemize}

The construction of $U_\gamma$ is similar to Pettie and Wang's~\cite{wang2023better} recent ``$\tau$-GRA'' 
cardinality estimators for \textsf{HyperLogLog} and \textsf{PCSA},
which aggregate statistics at each level $2^{-k}$ with weight $2^{-\tau k}$, called \emph{$\tau$-aggregation}. 
$U_\gamma$ can be considered as a $(-1/3)$-aggregation in \cite{wang2023better}'s framework.

\begin{remark}
It may seem unnatural to first construct an estimator $U_\gamma$ for $f_\gamma(\mathbf{x})^{1/3}$ and then  estimate $f_\gamma(\mathbf{x})$ with three i.i.d.~copies of $U_\gamma$. 
It would be simpler to set $U_\gamma' = \sum_{k\in \Z} (1-e^{i\gamma R_{e^{-k/m}}}) e^{k/m}$ or  $U_\gamma'' = \sum_{k\in \Z} (1-e^{i\gamma R_{e^{-k/m}}}) e^{\frac{1}{2}k/m}$, where by the same line of argument, one gets $\E U_\gamma'\propto f_\gamma(\mathbf{x})$ and $\E U_\gamma''\propto f_\gamma(\mathbf{x})^{1/2}$. However, we note that $\E U_\gamma'$ does not, in fact, exist, and $\E (U_\gamma'')^2$ also does not exist. Thus, for the estimator to have finite mean and variance, three copies are needed in this approach.
\end{remark}

\section{The \SymmetricPoissonTower{} Sketch}\label{sec:harmonic_ana}
Recall that a $\SymmetricPoisson(\lambda)$ random variable distributes as the difference of two independent Poisson random variables with rate $\lambda/2$. We first list a few more properties of symmetric Poissons.
\begin{lemma}\label{lem:symp_proper}
Let $Y\sim \SymmetricPoisson(\lambda)$.
\begin{itemize}
    \item $\E Y = 0$, $\E Y^2=\var Y = \lambda$, $\E Y^4=\lambda+3\lambda^2$, $\var Y^2 = \lambda +2\lambda^2$.
    \item $\E e^{izY} = e^{-\lambda (1-\cos(z))}$ for any $z\in \R$.
    \item $\pr(Y=0)=e^{-\lambda/2}$.
\end{itemize}
\end{lemma}

Following the discussion in \cref{sec:harmonic_intro,sec:choose_t}, we now formally define 
the \SymmetricPoissonTower{} sketch.

\begin{definition}
An \emph{abstract \SymmetricPoissonTower{}} is parameterized by a 
single integer $m$, which controls the variance of the estimates, 
and consists of an infinite vector $(X_j)_{j\in \Z}$ of \emph{cells}, initialized as $0^{\Z}$. For any $j\in \Z$,
\begin{align*}
    X_j = \sum_{v\in [n]}\mathbf{x}_v Z_{v,j},
\end{align*}
where $Z_{v,j} \sim \SymmetricPoisson(e^{-j/m})$. 
All $Z_{v,j}$s are independent and assigned by the random oracle. 
In other words, $\Update(v,y)$ is implemented as follows.
\begin{description}
    \item[$\Update(v,y)$:] 
    For each $j\in \Z$, set $X_j\gets X_j + Z_{v,j}y$, where $Z_{v,j}$s are independent $\mathrm{Poisson}(e^{-j/m})$ random variables.
\end{description}
The \emph{$[a,b)$-truncated \SymmetricPoissonTower} 
only stores $(X_j)_{j\in[a,b)}$.
\end{definition}

We will first analyze this abstract (infinite) \SymmetricPoissonTower{} and then prove that it suffices to store an 
$[a,b)$-truncated \SymmetricPoissonTower{} 
with $b-a = O(m\log(nmM))$ cells, 
at a cost of $O(1/\poly(nmM))$ to the bias of the estimates. See \cref{tab:notation} for notations.

\begin{table}[]
    \centering
    \begin{tabular}{c|c|c|c}
        notation & definition & name & note \\
        \hline
       $\norm{\mathbf{x}}_0$  & $\sum_{v\in[n]}\ind{\mathbf{x}_v\neq 0}$ &  support size & $\mathbf{x}\in \Z^n$ \\
       $\norm{\mathbf{x}}_p$  & $(\sum_{v\in[n]}|\mathbf{x}_v|^p)^{1/p}$ & $p$-norm & $\mathbf{x}\in \Z^n$, $p\in(0,\infty)$\\      
       $\norm{\mathbf{x}}_\infty$  & $\max_{v\in[n]}|\mathbf{x}_v|$ & maximum norm & $\mathbf{x}\in \Z^n$\\ 
       $f_\gamma(\mathbf{x})$  & $\sum_{v\in[n]}(1-\cos(\gamma \mathbf{x}_v))$ & $\gamma$-harmonic moment & $\mathbf{x}\in \Z^n,\gamma\in\R_+$\\ 
       $\var Z$  & $\E((Z-\E Z)\overline{(Z-\E Z)})$& variance & complex random variable $Z$\\       
    \end{tabular}
    \caption{Notations}
    \label{tab:notation}
\end{table}

\subsection{$(-1/3)$-Aggregation}
As discussed in \cref{sec:choose_t}, we will use the following statistic that selects suitable levels implicitly.
\begin{definition}[$(-1/3)$-aggregation]
For any $\gamma\in \R_+$, define
\begin{align*}
    U_{\gamma} = \sum_{k=-\infty}^\infty (1-e^{i\gamma X_k})e^{\frac{1}{3}k/m}.
\end{align*}
\end{definition}

To invoke Fubini's theorem when computing $\E U_\gamma$ and $\E U_\gamma^2$, we first prove the following lemma.
\begin{lemma}\label{lem:fubini}
For any $\gamma \in \R_+$ and $m\in\Z_+$,
\begin{itemize}
    \item $\sum_{k=-\infty}^\infty \E |1-e^{i\gamma X_k}|e^{\frac{1}{3}k/m}\leq 14m\norm{\mathbf{x}}_0^{1/3},$
    \item $\sum_{k=-\infty}^\infty \E |1-e^{i\gamma X_k}|^2e^{\frac{2}{3}k/m} \leq 44m\norm{\mathbf{x}}_0^{2/3}$.
\end{itemize}
\end{lemma}
\begin{proof}
Recall that $Z_{v,k}\sim \SymmetricPoisson(e^{-k/m})$ and thus by \cref{lem:symp_proper}, $\pr(Z_{v,k}=0)=e^{-e^{-k/m}/2}$.
Since $X_k = \sum_{v=1}^n Z_{v,k}\mathbf{x}_v$,
\begin{align*}
    \pr(X_k=0) \geq \pr(\forall v\in \supp(\mathbf{x}), Z_{v,k}=0) = e^{-e^{-k/m}\norm{\mathbf{x}}_0/2}.
\end{align*}
Note that $|1-e^{i\gamma X_k}|\leq 2$. Thus $\E |1-e^{i\gamma X_k}| \leq 2\pr(X_k\neq 0)\leq 2(1-e^{-e^{-k/m}\norm{\mathbf{x}}_0/2})$ and  $\E |1-e^{i\gamma X_k}|^2 \leq 4\pr(X_k\neq 0)\leq 4(1-e^{-e^{-k/m}\norm{\mathbf{x}}_0/2})$. Now we compute
\begin{align*}
    &\sum_{k=-\infty}^\infty \E |1-e^{i\gamma X_k}|e^{\frac{1}{3}k/m} \\
    &\leq 2 \sum_{k=-\infty}^\infty (1-e^{-e^{-k/m}\norm{\mathbf{x}}_0/2})e^{\frac{1}{3}k/m}\\
    &=2\sum_{k=-\infty}^\infty\eta_1(k/m;a,b,c) \intertext{where $\eta_1$ is defined in \cref{lem:eta} and $a=0,b=\norm{\mathbf{x}}_0/2,c=1/3$}
    &\leq 2m\int_{-\infty}^\infty \eta_1(x;a,b,c) \,dx + 2\left|m\int_{-\infty}^\infty \eta_1(x;a,b,c) -\sum_{k=-\infty}^\infty\eta_1(k/m;a,b,c)  \right|\\
    &\leq 2(-\Gamma(-1/3))m(\norm{\mathbf{x}}_0/2)^{1/3} + 2(3+\Gamma(2/3))(\norm{\mathbf{x}}_0/2)^{1/3}\tag{\cref{lem:eta}(\ref{item:lem:eta1},\ref{item:lem:eta6})}\\
    &\leq 14m\norm{\mathbf{x}}_0^{1/3}.
\end{align*}
Along the same lines, we have 
\begin{align*}
    &\sum_{k=-\infty}^\infty \E |1-e^{i\gamma X_k}|^2e^{\frac{1}{3}k/m} \\
    &\leq 4\sum_{k=-\infty}^\infty\eta_1(k/m;a,b,c) \intertext{where $a=0,b=\norm{\mathbf{x}}_0/2,c=2/3$}
    &\leq 4(-\Gamma(-2/3))m(\norm{\mathbf{x}}_0/2)^{2/3} + 4(6+\Gamma(1/3))(\norm{\mathbf{x}}_0/2)^{2/3}\tag{\cref{lem:eta}(\ref{item:lem:eta1},\ref{item:lem:eta6})}\\
    &\leq 44m\norm{\mathbf{x}}_0^{2/3}.
\end{align*}
\end{proof}

We now compute the mean and variance of $U_{\gamma}$. Note that we will use the Gamma function $\Gamma(x)$, which is negative over $(-1,0)$ and positive over $(0,1)$. See \cref{fig:gamma_function}. 

\begin{figure}
    \centering
    \includegraphics[width=0.5\linewidth]{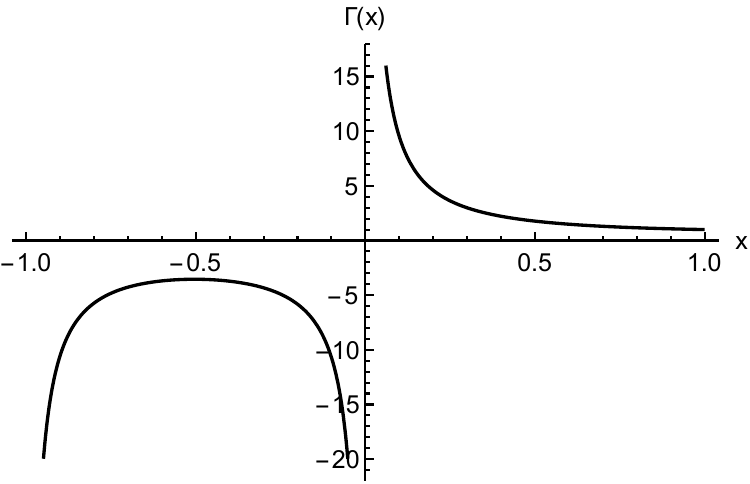}
    \caption{Gamma function $\Gamma(x)$ over $(-1,0)\cup (0,1)$}
    \label{fig:gamma_function}
\end{figure}

\begin{lemma}[mean of $U_\gamma$]\label{lem:E_u}
For any $\gamma \in \R_+$,
\begin{align*}
m^{-1}\E U_{\gamma} = 
(-\Gamma(-1/3) + O(m^{-1}))f_\gamma(\mathbf{x})^{1/3},
\end{align*}
as $m\to \infty$.
\end{lemma}
\begin{proof}
By \cref{lem:fubini}, we invoke Fubini's theorem and exchange the order of expectation and sum.
\begin{align*}
    \E U_\gamma &= \sum_{k=-\infty}^\infty (1- \E e^{i\gamma X_k})e^{\frac{1}{3}k/m}
    \intertext{ By \cref{lem:observation},}
    &= \sum_{k=-\infty}^\infty (1- e^{-e^{-k/m} f_\gamma(\mathbf{x})})e^{\frac{1}{3}k/m}\\
    &= \sum_{k=-\infty}^\infty \eta_1(k/m;a,b,c),
\end{align*}
where $a=0,b=f_\gamma(\mathbf{x}),c=1/3$. By \cref{lem:eta}, we know
\begin{align*}
    \left|\E U_\gamma - mf_\gamma(\mathbf{x})^{1/3}(-\Gamma(-1/3))\right| &\leq (3+\Gamma(2/3))f_\gamma(\mathbf{x})^{1/3}.
\end{align*}
\end{proof}

\begin{lemma}[variance of $U_\gamma$]\label{lem:V_u}
For any $\gamma\in \R_+$,
\begin{align*}
    m^{-1}\var U_\gamma = 2^{2/3}(-\Gamma(-2/3))(1+O(m^{-1}))f_\gamma(\mathbf{x})^{2/3},
\end{align*}
as $m\to \infty$.
\end{lemma}
\begin{proof}
We first check the integrability.
\begin{align*}
    &\sum_{k=-\infty}^\infty\sum_{j=-\infty}^\infty \E |1-e^{i\gamma X_k}||1-e^{i\gamma X_j}|e^{\frac{1}{3}(k+j)/m}\\
    &\leq \sum_{k=-\infty}^\infty \E |1-e^{i\gamma X_k}|^2e^{\frac{2}{3}k/m}+(\sum_{k=-\infty}^\infty \E |1-e^{i\gamma X_k}|e^{\frac{1}{3}k/m})^2\\
    &\leq 240 m^2 \norm{\mathbf{x}}_0^{2/3}<\infty \tag{by \cref{lem:fubini}}.
\end{align*}
Then, by Fubini's theorem, we have
\begin{align*}
    \var U_\gamma = \E U_\gamma\overline{U_\gamma} -\E U_\gamma \E \overline{U_\gamma}  &= \sum_{k=-\infty}^\infty\sum_{j=-\infty}^\infty \E (1-e^{i\gamma X_k})(1-e^{-i\gamma X_j})e^{\frac{1}{3}(k+j)/m}\\
    &\quad -\sum_{k=-\infty}^\infty\sum_{j=-\infty}^\infty \E (1-e^{i\gamma X_k})\E (1-e^{-i\gamma X_j})e^{\frac{1}{3}(k+j)/m}\\
    &= \sum_{k=-\infty}^\infty (1-|\E e^{i\gamma X_k}|^2) e^{\frac{2}{3}k/m}\\
    &= \sum_{k=-\infty}^\infty (1-e^{-2e^{-k/m}[\mathrm{x};\gamma]}) e^{\frac{2}{3}k/m} \\
    &= \eta_1(k/m; 0, 2f_\gamma(\mathbf{x}),2/3).
\end{align*}
 By \cref{lem:eta}, we know
 \begin{align*}
\left|\var U_\gamma-m (2f_\gamma(\mathbf{x}))^{2/3}(-\Gamma(-2/3))\right|\leq (6+\Gamma(1/3))(2f_\gamma(\mathbf{x}))^{2/3}
 \end{align*}
\end{proof}

\subsection{Triple \SymmetricPoissonTower}

As in \cref{sec:choose_t}, the final trick is to store three independent copies of the \SymmetricPoissonTower{} and obtain three i.i.d.~copies of $U_{\gamma}$, denoted as $U_{\gamma}^{(1)}$, $U_{\gamma}^{(2)}$, $U_{\gamma}^{(3)}$.  Define $V_{\gamma}$ as
\begin{align*}
    V_{\gamma} =  (-\Gamma(-1/3))^{-3}m^{-3}U_{\gamma}^{(1)}U_{\gamma}^{(2)}U_{\gamma}^{(3)}.
\end{align*}

\begin{theorem}[triple \SymmetricPoissonTower]\label{thm:abstract}
Let $X=(X_k^{(j)})_{k\in \Z,j=1,2,3}$ be a triple \SymmetricPoissonTower. For any $\gamma \in \R_+$, let
\begin{align*}
    V_\gamma &= (-\Gamma(-1/3))^{-3}m^{-3} \prod_{j=1}^3\left(\sum_{k=-\infty}^\infty (1-e^{i\gamma X_k^{(j)}})e^{\frac{1}{3}k/m}\right)
\end{align*}
Then
\begin{align*}
    \E V_\gamma &= (1+O(m^{-1}))f_\gamma(\mathbf{x})\\
    \var V_\gamma &= (m^{-1}+O(m^{-2}))\cdot \frac{3\cdot 2^{2/3}(-\Gamma(-2/3))}{(-\Gamma(-1/3))^2} \cdot f_\gamma(\mathbf{x})^{2}.
\end{align*}
\end{theorem}
\begin{proof}
Recall that $ (-\Gamma(-1/3))^{3}V_{\gamma} = m^{-3}U_{\gamma}^{(1)}U_{\gamma}^{(2)}U_{\gamma}^{(3)}$, where $U^{(j)_\gamma}$s are i.i.d.~copies of $U_\gamma$. We first compute the mean. By \cref{lem:E_u},
    \begin{align*}
        &\E V_\gamma= \left((-\Gamma(-1/3))^{-1}m^{-1}\E U_\gamma\right)^3 \\
        &=\left((1+O(m^{-1})) f_\gamma(\mathbf{x})^{1/3}\right)^3\\
        &=  (1+O(m^{-1})) f_\gamma(\mathbf{x}).
    \end{align*}
We now compute the variance. 
    \begin{align*}
       (-\Gamma(-1/3))^6 \E V_\gamma\overline{V_\gamma} &=\left(m^{-2}\E U_\gamma\overline{U_\gamma}\right)^3\\
       &= \left(m^{-2}\var U_\gamma + m^{-2}\E U_\gamma\E \overline{U_\gamma}\right)^3\\       
      (-\Gamma(-1/3))^6  \E V_\gamma\E\overline{V_\gamma}&=\left(m^{-2}\E U_\gamma\E\overline{U_\gamma}\right)^3\\
       (-\Gamma(-1/3))^6 \var V_\gamma &= \left(m^{-2}\var U_\gamma + m^{-2}\E U_\gamma\E \overline{U_\gamma}\right)^3- \left(m^{-2}\E U_\gamma\E\overline{U_\gamma}\right)^3\\
        &= m^{-3}\left(m^{-1}\var U_\gamma\right)^3+3m^{-2}\left(m^{-1}\var U_\gamma\right)^2\left(m^{-2}\E U_\gamma\E \overline{U_\gamma}\right)\\
        &\quad+3m^{-1}\left(m^{-1}\var U\right)\left(m^{-2}\E U_\gamma\E \overline{U_\gamma}\right)^2.
        \intertext{Note that by \cref{lem:V_u}, $m^{-1}\var U_\gamma =(2^{2/3}(-\Gamma(-2/3))+O(m^{-1}))f_\gamma(\mathbf{x})^{2/3}$ and by \cref{lem:E_u}, $|m^{-1}\E U_\gamma|=(-\Gamma(-1/3)+O(m^{-1})) f_\gamma(\mathbf{x})^{1/3}$. Thus the last term dominates and we have}
        &=3m^{-1}\left(m^{-1}\var U_\gamma)(m^{-2}\E U_\gamma\E \overline{U_\gamma}\right)^2\left(1+O(m^{-1})\right)\\
        &= 3m^{-1}\left(2^{2/3}(-\Gamma(-2/3))\right)f_\gamma(\mathbf{x})^{2/3}(-\Gamma(-1/3))^4 f_\gamma(\mathbf{x})^{4/3}\left(1+O(m^{-1})\right).
    \end{align*}
    We conclude that
    \begin{align*}
        \var V_\gamma &=  (m^{-1}+O(m^{-2}))\cdot \frac{3\cdot 2^{2/3}(-\Gamma(-2/3))}{(-\Gamma(-1/3))^2} \cdot f_\gamma(\mathbf{x})^{2}.
    \end{align*}
\end{proof}

\subsection{The Truncated Tower}
We now consider a finite truncation $(X_k)_{a\leq k < b}$. The discussion in \cref{sec:choose_t} suggests that we only need to store an interval of cells in order to estimate the $\gamma$-harmonic moment $f_\gamma(\mathbf{x})$. Define the $[a,b)$-truncated $(-1/3)$-aggregation as
\begin{align*}   
U_{\gamma;a,b}=\sum_{k=a}^{b-1} \left(1-e^{i\gamma X_k}\right)e^{\frac{1}{3}k/m},
\end{align*}
The statistic $U_{\gamma;a,b}$ can be obtained by maintaining only $(b-a)$ cells. We now formally show that the contribution to the first and second moments being truncated decays exponentially in both directions.

\begin{lemma}\label{lem:ubc_bounds}
    Let  $A=U_\gamma$, $B=U_{\gamma;-\infty,a}$ and $C=U_{b,\infty}$. Define $\norm{\mathbf{x}}_0=\sum_{v\in[n]}\ind{\mathbf{x}_v\neq 0}$.  Then
    \begin{multicols}{2}
    \begin{enumerate}
        \item $\E |A| \leq \poly(\norm{\mathbf{x}}_0m)$.
        \item $|B| \leq  e^{\frac{1}{3}a/m}\poly(\norm{\mathbf{x}}_0m)$.
        \item $\E |C| \leq e^{-\frac{2}{3}b/m}\poly(\norm{\mathbf{x}}_0m)$.
        \item $\E |A|^2 \leq \poly(\norm{\mathbf{x}}_0m)$.
        \item $ |B|^2 \leq e^{\frac{2}{3}a/m}\poly(\norm{\mathbf{x}}_0 m)$.
        \item $\E |C|^2 \leq e^{-\frac{4}{3}b/m}\poly(\norm{\mathbf{x}}_0m)$.
        \item $\E |A||B| \leq e^{\frac{1}{3}a/m}\poly(\norm{\mathbf{x}}_0m)$.
        \item $\E|B||C| \leq e^{\frac{1}{3}a/m-\frac{2}{3}b/m}\poly(\norm{\mathbf{x}}_0m)$.
        \item $\E|C||A| \leq e^{-\frac{2}{3}b/m}\poly(\norm{\mathbf{x}}_0m)$.
        \item[]
    \end{enumerate}
    \end{multicols}
\end{lemma}
\begin{proof}
\begin{enumerate}
    \item Follow from \cref{lem:fubini} by noting that $\E |\sum_{k\in \Z}(1-e^{i\gamma X_k})e^{\frac{1}{3}k/m}|\leq \sum_{k\in \Z}\E |1-e^{i\gamma X_k}|e^{\frac{1}{3}k/m}$.
    \item $|B|=|\sum_{k=-\infty}^{a-1}(1-e^{i\gamma X_k})e^{\frac{1}{3}k/m}|\leq 2\sum_{k=-\infty}^{a-1}e^{\frac{1}{3}k/m}=e^{\frac{1}{3}a/m}\cdot O(m)$.
    \item $\E|C|\leq \sum_{k=b}^\infty \E|1-e^{i\gamma X_k}|e^{\frac{1}{3}k/m}$. As in \cref{lem:fubini}, $\E |1-e^{i\gamma X_k}|\leq 2(1-e^{-e^{-k/m}\norm{\mathbf{x}}_0/2})\leq e^{-k/m}\norm{\mathbf{x}}_0$. Thus $\E|C|\leq  \sum_{k=b}^\infty \norm{\mathbf{x}}_0e^{-\frac{2}{3}k/m}=\norm{\mathbf{x}}_0 e^{-\frac{2}{3}b/m}\cdot O(m)$.
    \item Proved in \cref{lem:V_u}.
    \item By part 2.
    \item As in \cref{lem:fubini}, we have $\E |1-e^{i\gamma X_k}|^2\leq 4(1-e^{-e^{-k/m}\norm{\mathbf{x}}_0/2})\leq 2 e^{-k/m}\norm{\mathbf{x}}_0$ and for $j\neq k$, $\E |1-e^{i\gamma X_k}||1-e^{i\gamma X_j}|=\E |1-e^{i\gamma X_k}|\E |1-e^{i\gamma X_j}|\leq e^{-(k+j)/m}\norm{\mathbf{x}}_0^2$. Thus 
    \begin{align*}
        \E |C|^2 &\leq \sum_{j\geq b}\sum_{k\geq b} \E|1-e^{i\gamma X_j}||1-e^{i\gamma X_k}|e^{\frac{j+k}{3}/m}\\
        &\leq  \sum_{j\geq b}\sum_{k\geq b} e^{-(j+k)/m}\norm{\mathbf{x}}_0^2e^{\frac{j+k}{3}/m}\\
        &=\norm{\mathbf{x}}_0^2(\sum_{j\geq b}e^{-\frac{2}{3}j/m})^2\\
        &= e^{-\frac{4}{3}b/m}\poly(\norm{\mathbf{x}}_0m).
    \end{align*}
    \item By parts 1 and 2.
    \item By parts 2 and 3.
    \item Follows from parts 4 and 6. Note that $\E |A||C|\leq \sqrt{\E|A|^2 \E |C|^2}$ by Cauchy-Schwartz.
\end{enumerate}
\end{proof}

\begin{theorem}[truncated harmonic estimator]\label{thm:truncation}
For any $\gamma>0$, define the truncated estimator
    \begin{align*}        V_{\gamma;a,b}&={m^{-3}(-\Gamma(-1/3))^{-3}}{\displaystyle\prod_{j=1}^3\left(\sum_{k= a}^{b-1}\left(1-e^{i\gamma X_k^{(j)}}\right)e^{\frac{1}{3}k/m}\right)}
    \end{align*}

For any $a\leq 0\leq b$,
 \begin{align*}
     \left|\E V_{\gamma;a,b} - \E V_{\gamma}\right| &\leq (e^{\frac{1}{3}a/m}+e^{-\frac{2}{3}b/m})\poly(\norm{\mathbf{x}}_0m)\\
     \left|\var V_{\gamma;a,b} - \var V_{\gamma}\right| &\leq (e^{\frac{1}{3}a/m}+e^{-\frac{2}{3}b/m})\poly(\norm{\mathbf{x}}_0m).
 \end{align*}
\end{theorem}
\begin{proof}
Let the $A^{(j)}$, $B^{(j)}$, and $C^{(j)}$ be i.i.d.~copies of $A$, $B$, and $C$, defined in \cref{lem:ubc_bounds}.
Recall that the infinite \SymmetricPoissonTower{} 
has estimator $$V_\gamma=\frac{A^{(1)}A^{(2)}A^{(3)}}{m^3(-\Gamma(-1/3))^3}$$ and the truncated \SymmetricPoissonTower{} has estimator $$V_{\gamma;a,b}=\frac{\left(A^{(1)}-B^{(1)}-C^{(1)}\right)\left(A^{(2)}-B^{(2)}-C^{(2)}\right)\left(A^{(3)}-B^{(3)}-C^{(3)}\right)}{m^3(-\Gamma(-1/3))^3}.$$
We first bound their difference in the mean.
\begin{align*}
    &|\E V_\gamma-\E V_{\gamma;a,b}|\\
    &=\left|(\E (A-B-C))^3 - (\E A)^3\right| \\
    &\leq 3 (\E |A|)^2(\E |B|+\E |C|)+3(\E |A|)(\E |B|+\E |C|)^2 + (\E |B|+\E |C|)^3 
    \intertext{According to \cref{lem:ubc_bounds}, we know $\E|B|\leq e^{\frac{1}{3}a/m}\poly(\norm{\mathbf{x}}_0 m)$, $\E|C|\leq e^{-\frac{2}{3}b/m}\poly(\norm{\mathbf{x}}_0 m)$, and $\E|A|\leq \poly(\norm{\mathbf{x}}_0 m)$. Since we assume $a\leq 0\leq b$, we have}
    &\leq (e^{\frac{1}{3}a/m}+e^{-\frac{2}{3}b/m})\poly(\norm{\mathbf{x}}_0m),
\end{align*}
 Now we bound the difference in the second moment.
\begin{align*}
    &\left|\E |V_\gamma|^2-\E |V_{\gamma;a,b}|^2\right|\\
    &=\left|\left(\E |A-B-C|^2\right)^3 - \left(\E |A|^2\right)^3\right|\\
    &=\left|\left(\E (|A|^2+|B|^2+|C|^2-A\overline{B}-A\overline{C}-B\overline{A}+B\overline{C}-C\overline{A}+C\overline{B}))^3 - (\E |A|^2\right)^3\right|,
    \intertext{where the dominating term $(\E |A|^2)^3$ is canceled and we have }
    &=O((\E|A|^2)^2 \cdot (\E|A||B| + \E|A||C|))
    \intertext{Similarly by \cref{lem:ubc_bounds}, we know $\E|A|^2=\poly(\norm{\mathbf{x}}_0m)$, $\E|A||B|=e^{\frac{1}{3}a/m}\poly(\norm{\mathbf{x}}_0m)$, and $\E|A||C|=e^{-\frac{2}{3}b/m}\poly(\norm{\mathbf{x}}_0m)$. We thus have}
    &\leq (e^{\frac{1}{3}a/m}+e^{-\frac{2}{3}b/m})\poly(\norm{\mathbf{x}}_0m).
\end{align*}
We have proved first and second moment have vanishing differences, which implies the difference in variance is also $(e^{\frac{1}{3}a/m}+e^{-\frac{2}{3}b/m})\poly(\norm{\mathbf{x}}_0m)$.
\end{proof}

\section{Analysis of Combined Estimators}\label{sec:ana_comb}
Let $\{V_{\gamma;a,b}\}_{\gamma>0}$ be the truncated harmonic estimators in \cref{thm:truncation}. We just showed that the mean of $V_{\gamma;a,b}$ can be made arbitrarily close to the mean of $V_\gamma$ by picking small enough $a$ and large enough $b$. Choosing $a=-\Omega(m\log(mn))$ and $b=\Omega(m\log(mn))$, the difference $|\E V_{\gamma;a,b} - \E V_{\gamma}|$ becomes $O(1/\poly(mn))$. For a function $f(x)=\int_0^\infty(1-\cos(\gamma x))\,\nu(d\gamma)$, to estimate the $f$-moment one may try to combine the individual estimates as discussed in \cref{sec:combine}. However, note that
\begin{align*}
    \E \int_0^\infty V_{\gamma;a,b} \,\nu(d\gamma)&=\int_0^\infty (\E V_{\gamma}+O(1/\poly(mn))) \,\nu(d\gamma) \tag{\cref{thm:truncation}}\\
    &=\int_0^\infty ((1+O(m^{-1}))f_\gamma(\mathbf{x})+O(1/\poly(mn))) \,\nu(d\gamma) \tag{\cref{thm:abstract}}\\
    &=(1+O(m^{-1}))f(\mathbf{x}) + O(1/\poly(mn))\cdot \int_0^\infty \,\nu(d\gamma),
\end{align*}
where in general $\int_0^\infty \,\nu(d\gamma)$ can be \emph{infinite}. For example, to estimate the $L_1$-moment, we have $f(x)=|x|=\int_0^\infty (1-\cos(\gamma x)) \frac{2}{\pi \gamma^2}\, d\gamma$. For this function $f$, we have $\int_0^\infty  \frac{2}{\pi \gamma^2}\, d\gamma =\infty$. Recall that for $\int_0^\infty (1-\cos(\gamma x))\,\nu(d\gamma)$ to be well-defined, one only needs $\int_0^\infty \min(\gamma^2,1)\,\nu(d\gamma)$ to be finite. The reason that the attempt above fails is that, when $\gamma$ is near 0, the density $\nu(d\gamma)$ is unbounded and any non-zero additive bias is too large for the combined estimator. We thus need to split $f$ into two parts $f_{<\zeta}(x)=\int_0^\zeta (1-\cos(\gamma x))\,\nu(d\gamma)$ and $f_{\geq\zeta}(x)=\int_\zeta^\infty (1-\cos(\gamma x))\,\nu(d\gamma)$ for estimation, where $\zeta$ is to be determined later.

Before we proceed, we first define a quantity that measures the ``size'' of the function $f$.
\begin{definition}[function constant $C_f$]
For any $f=\int_0^\infty (1-\cos(\gamma x))\,\nu(d\gamma)$ where $\nu$ is a positive measure, define $C_f = \int_0^\infty \min(\gamma^2,1)\,\nu(d\gamma)$. Note that  for any $x\in\R$, $f(x)$ is finite if and only if $C_f<\infty$.
\end{definition}

\subsection{Estimation of the $f_{<\zeta}$-Moment}\label{sec:L2}
We first show that, when $\zeta$ is small enough, the $f_{<\zeta}$-moment can be approximated by the $L_2$-moment $\norm{\mathbf{x}}_2^2$.
\begin{lemma}\label{lem:small_part}
Recall that $f_{<\zeta}(x)=\int_0^{\zeta} (1-\cos (\gamma x)) \, \nu(d\gamma) $. The $f_{<\zeta}$-moment of $\mathbf{x}\in \Z^n$ can be bounded as follows.
\begin{align*}
    \left(1-\frac{\zeta^2\norm{\mathbf{x}}_2^2}{12}\right)\cdot \frac{1}{2}\norm{\mathbf{x}}_2^2\int_0^{\zeta} \gamma^2 \, \nu(d\gamma)\leq f_{<\zeta}(\mathbf{x})\leq \frac{1}{2}\norm{\mathbf{x}}_2^2\int_0^{\zeta} \gamma^2 \, \nu(d\gamma).
\end{align*}
\end{lemma}
\begin{proof}
We first note the following approximation of $(1-\cos(x))$. For any $x\in \R$, $$x^2/2-x^4/24 \leq 1-\cos(x) \leq x^2/2.$$
On the one hand, we have
\begin{align*}
    \sum_{v\in[n]}\int_0^{\zeta} (1-\cos (\gamma \mathbf{x}_v)) \, \nu(d\gamma) &\leq \sum_{v\in[n]}\int_0^{\zeta} \gamma^2 \mathbf{x}_v^2/2 \, \nu(d\gamma)\\
    &=\frac{1}{2}\norm{\mathbf{x}}_2^2\int_0^{\zeta} \gamma^2 \, \nu(d\gamma).
\end{align*}
On the other hand,
\begin{align*}
\sum_{v\in[n]}\int_0^{\zeta} (1-\cos (\gamma \mathbf{x}_v)) \, \nu(d\gamma) &\geq \sum_{v\in[n]}\int_0^{\zeta} (\gamma^2 \mathbf{x}_v^2/2-\gamma^4 \mathbf{x}_v^4/24) \, \nu(d\gamma)\\  
&= \frac{1}{2}\norm{\mathbf{x}}_2^2\int_0^{\zeta} \gamma^2 \, \nu(d\gamma) - \frac{1}{24}\norm{\mathbf{x}}_4^4\int_0^{\zeta} \gamma^4 \, \nu(d\gamma)\\
&\geq \frac{1}{2}\norm{\mathbf{x}}_2^2\int_0^{\zeta} \gamma^2 \, \nu(d\gamma) - \frac{1}{24}\norm{\mathbf{x}}_4^4 \zeta^2\int_0^{\zeta} \gamma^2 \, \nu(d\gamma).
\end{align*}
Finally we note that $\norm{\mathbf{x}}_4\leq \norm{\mathbf{x}}_2$.
\end{proof}
Thus by choosing $\zeta=O(1/\poly(m\norm{\mathbf{x}}_2))$, $\frac{1}{2}\norm{\mathbf{x}}_2^2\int_0^{\zeta} \gamma^2 \, \nu(d\gamma)$ will be a $(1\pm O(m^{-1}))$-approximation of the $f_{<\zeta}$-moment. It suffices now to approximate the $L_2$-moment $\norm{\mathbf{x}}_2^2$. The classic $L_2$-moment can be estimated with a separate AMS  sketch \cite{AlonMS99}. Nevertheless, to demonstrate the universality of the \SymmetricPoissonTower{} 
sketch, we will use its cells $X_0,\ldots,X_{m-1}$ to estimate $L_2$, which are \emph{not} independent from the harmonic estimators $\{V_{\gamma;a,b}\}$.
\begin{lemma}[$L_2$-moment estimator]\label{lem:L2}
Let $(X_k)_{a\leq k<b}$ be a truncated symmetric Poisson tower with parameter $m$ where $a\leq 0$ and $b\geq m$. Then we define the $L_2$-estimator,
\begin{align*}
    \Psi &= \frac{1}{m} \sum_{k=0}^{m-1} X_k^2e^{k/m}.
\end{align*}
We have
\begin{align*}
    \E \Psi &= \norm{\mathbf{x}}_2^2\\
    \var \Psi &\leq (3+e)m^{-1}\norm{\mathbf{x}}_2^4.
\end{align*}
\end{lemma}
\begin{proof}
    Recall that $X_k=\sum_{v\in[n]} Z_{v,k}\mathbf{x}_v$ where $Z_{v,k}\sim \SymmetricPoisson(e^{-k/m})$. By \cref{lem:symp_proper}, we know $\E Z_{v,k}=0$, $\E Z_{v,k}^2=e^{-k/m}$, and $\E Z_{v,k}^4=e^{-k/m}+3e^{-2k/m}$. Thus similar to the AMS sketch,
    \begin{align*}
        \E X_k^2 = \sum_{v\in[n]} \mathbf{x}_v^2 \E Z_{v,k}^2=e^{-k/m}\norm{\mathbf{x}}_2^2.
    \end{align*}
    Thus $\E \Psi = \norm{\mathbf{x}}_2^2$. Then
    \begin{align*}
        \E X_k^4 &=  \sum_{v\in[n]} \mathbf{x}_v^4 \E Z_{v,k}^4 + 6\sum_{s<t} \mathbf{x}_s^2\mathbf{x}_t^2 \E Z_{v,s}^2\E Z_{v,t}^2\\
        &\leq (e^{-k/m}+3e^{-2k/m})\norm{\mathbf{x}}_2^4
    \end{align*}Now
    \begin{align*}
        \E \Psi^2 &\leq \frac{1}{m^2}\sum_{k=0}^{m-1}\E X_k^4 e^{2k/m} + (\frac{1}{m} \sum_{k=0}^{m-1} \E X_k^2e^{k/m})^2\\
        &= \frac{\norm{\mathbf{x}}_2^4}{m^2} \sum_{k=0}^{m-1}(e^{k/m}+3)+ \norm{\mathbf{x}}_2^4\\
        &\leq \frac{\norm{\mathbf{x}}_2^4}{m} (3+e)+ \norm{\mathbf{x}}_2^4.
    \end{align*}
Thus $\var \Psi\leq \E \Psi^2-(\E \Psi)^2\leq  (3+e)m^{-1}\norm{\mathbf{x}}_2^4$.
\end{proof}

\subsection{Estimation of the $f_{\geq \zeta}$-Moment}\label{sec:est_fgeps}
We proceed as in \cref{sec:combine}. Let $\{V_{\gamma;a,b}\}_{\gamma\geq \zeta}$ be the set of harmonic estimators generated from the truncated triple \SymmetricPoissonTower{} in \cref{thm:truncation}. For $f_{\geq \zeta}(x)=\int_\zeta^\infty (1-\cos(\gamma x))\,\nu(d\gamma)$, define
\begin{align*}
    \phi_{f_{\geq \zeta};a,b} &= \int_\zeta^\infty V_{\gamma;a,b} \,\nu(d\gamma).
\end{align*}
By \cref{thm:truncation}, we know $\E V_{\gamma;a,b} = \E V_\gamma + O(\delta)$ where $\delta =(e^{\frac{1}{3}a/m}+e^{-\frac{2}{3}b/m})\poly(\norm{\mathbf{x}}_0m)$.

We have
\begin{align*}
    \E \phi_{f_{\geq \zeta};a,b} &=  \int_\zeta^\infty \E V_{\gamma;a,b} \,\nu(d\gamma) \\
    &= \int_\zeta^\infty ((1+O(m^{-1}))f_\gamma(\mathbf{x}) + O(\delta)) \,\nu(d\gamma) \\
    &=(1+O(m^{-1})) \int_\zeta^\infty  f_\gamma(\mathbf{x}) \,\nu(d\gamma) + O(\delta) \int_\zeta^\infty\,\nu(d\gamma)
    \intertext{note that since $\zeta<1$, we have $\int_\zeta^\infty\,\nu(d\gamma)\leq \zeta^{-2}\int_0^\infty \min(\gamma^2,1) \,\nu(d\gamma) =\zeta^{-2}C_f$}
    &=(1+O(m^{-1})) \int_\zeta^\infty  f_\gamma(\mathbf{x}) \,\nu(d\gamma) + O(\delta) \zeta^{-2}C_f.
\end{align*}
The variance can be bounded similarly and we will leave it in the proof of \cref{thm:frequency_combination}.

\subsection{Combination}
We now prove the main technical theorem.
\begin{theorem}[``frequency-domain'' estimator]\label{thm:frequency_combination}
Let $\mathbf{x}\in \Z^n$ be the input frequency vector and  $(X_{k}^{(j)})_{a\leq k<b,j=1,2,3}$ be a $[a,b)$-truncated triple \SymmetricPoissonTower{} with parameter $m$. For any function $f(x)=\sigma^2 x^2/2+\int_0^\infty (1-\cos(\gamma x))\,\nu(d\gamma)$ where $\nu$ is a positive measure with $C_f = \int_0^\infty \min(\gamma^2,1)\,\nu(d\gamma)$. Denote the estimator for $f$ as
\begin{align*}
    \phi_{f;a,b} &= \frac{1}{2}\left(\sigma^2+\int_0^\zeta \gamma^2\,\nu(d\gamma)\right)\frac{\sum_{j=1}^3\sum_{k=0}^{m-1}(X_k^{(j)})^2e^{k/m}}{3m}\\
    &\quad + \int_\zeta^\infty \frac{\prod_{j=1}^3\left(\sum_{k=a}^{b-1}(1-e^{i\gamma X_k^{(j)}})e^{\frac{1}{3}k/m}\right)}{m^3(-\Gamma(-1/3))^3} \,\nu(d\gamma).
\end{align*}
Set $\zeta=O(1/\poly(\norm{\mathbf{x}}_2m))$. If $a=-\Omega(m\log(\norm{\mathbf{x}}_2m))$ and $b=\Omega(m\log(\norm{\mathbf{x}}_2m))$, then $\phi_{f;a,b}$ is a $(1\pm O(m^{-1}))$-approximate of the $f$-moment with negligible bias.
\begin{align*}
    \E \phi_{f;a,b} &= (1+O(m^{-1}))\sum_{v\in[n]}f(\mathbf{x}_v) + O(\frac{C_f}{\poly(\norm{\mathbf{x}}_2 m)})\\
    \var \phi_{f;a,b} &\leq O(m^{-1}) \cdot \left(\sum_{v\in[n]}f(\mathbf{x}_v)\right)^2+ O(\frac{C_f^2}{\poly(\norm{\mathbf{x}}_2 m)}).\\
\end{align*}
\end{theorem}
\begin{proof}
For simplicity, we write $\phi_{f;a,b}$ as $\alpha A + B$ where $\alpha = \frac{1}{2}(\sigma^2+\int_0^\zeta \gamma^2\,\nu(d\gamma))$, $A = \frac{\sum_{j=1}^3\sum_{k=0}^{m-1}(X_k^{(j)})^2e^{k/m}}{3m}$, and $B=\int_\zeta^\infty \frac{\prod_{j=1}^3(\sum_{k=a}^{b-1}(1-e^{i\gamma X_k^{(j)}})e^{\frac{1}{3}k/m})}{m^3(-\Gamma(-1/3))^3} \,\nu(d\gamma)$. By \cref{lem:L2}, we know
\begin{align*}
    \E A &= \norm{\mathbf{x}}_2^2\\
    \var A &\leq (1+e/3)m^{-1}\norm{\mathbf{x}}_2^4.
\end{align*}
With \cref{lem:small_part}, we know 
\begin{align*}
    \E \alpha A &= (1+O(m^{-1}))f_{<\zeta}(\mathbf{x}) \\
    \var (\alpha A) &\leq (1+e/3)m^{-1}\left(f_{<\zeta}(\mathbf{x})\right)^2(1+O(m^{-1}))
\end{align*}
Since $\mathbf{x}\in \Z^n$, we always have $\norm{\mathbf{x}}_0\leq \norm{\mathbf{x}}_2^2$. By the calculation in \cref{sec:est_fgeps} and the assumptions on $\zeta,a,b$, we have
\begin{align*}
    \E B &= (1+O(m^{-1})) \sum_{v\in[n]}f_{\geq \zeta}(\mathbf{x}_v) + C_f\cdot O(1/\poly(\norm{\mathbf{x}}_2m)) 
    \intertext{Similarly, for the variance we have}
    \var B &\leq \left(\int_\zeta^\infty \sqrt{\var V_{\gamma;a,b}}\,\nu(\gamma)\right)^2 \\
    &\leq \left(\int_\zeta^\infty \sqrt{\var V_\gamma + O(1/\poly(\norm{\mathbf{x}}_2m))}\,\nu(\gamma)\right)^2 \tag{\cref{thm:truncation}}\\
    &\leq \left(\int_\zeta^\infty \sqrt{(m^{-1}+O(m^{-2}))\cdot C \cdot f_\gamma(\mathbf{x})^2 + O(1/\poly(\norm{\mathbf{x}}_2m))}\,\nu(\gamma)\right)^2 \tag{\cref{thm:abstract}}
    \intertext{note that $f_\gamma(\mathbf{x})\leq 2\norm{\mathbf{x}}_0\leq 2\norm{\mathbf{x}}_2^2$ and $\int_\zeta^{\infty}\,\nu(d\gamma)\leq \zeta^{-2}\int_\zeta^{\infty}\min(\gamma^2,1)\,\nu(d\gamma) = \zeta^{-2}C_f $.}
    &\leq \left(\int_\zeta^\infty \sqrt{(m^{-1}+O(m^{-2}))\cdot C \cdot f_\gamma(\mathbf{x})^2} + \sqrt{O(1/\poly(\norm{\mathbf{x}}_2m))}\,\nu(\gamma)\right)^2 \\
    &\leq \left( \sqrt{(m^{-1}+O(m^{-2}))\cdot C}f_{\geq\zeta}(\mathbf{x}) + \sqrt{O(1/\poly(\norm{\mathbf{x}}_2m))}\zeta^{-2}C_f\right)^2 
\end{align*}
Note that we have $1+e/3<2$ and $C=\frac{3\cdot 2^{2/3}(-\Gamma(-2/3))}{(-\Gamma(-1/3))^2}<2$. 
We conclude that
\begin{align*}
    \E(\alpha A + B) &= (1+O(m^{-1}))f(\mathbf{x})+ C_f\cdot O(1/\poly(\norm{\mathbf{x}}_2m))\\
    \var(\alpha A + B) &\leq (\sqrt{\var(\alpha A)}+\sqrt{\var B})^2\\
    &\leq \left(\sqrt{2(m+O(m^{-2}))}f(\mathbf{x}) + O(1/\poly(\norm{\mathbf{x}}_2m))\cdot C_f\right)^2
    \intertext{We know  $f_{\geq\zeta}(\mathbf{x})\leq \sum_{v\in[n]}\int_\zeta^\infty 2 \,\nu(d\gamma)= O(\zeta^{-2}C_f\norm{\mathbf{x}}_0) $ and $f_{<\zeta}(\mathbf{x})\leq \frac{1}{2}\norm{\mathbf{x}}_2^2 C_f$ (\cref{lem:small_part}). Therefore $f(\mathbf{x})=O(\zeta^{-2}C_f\norm{\mathbf{x}}_0)$. We have assumed $\zeta=O(1/\poly(\norm{\mathbf{x}}_2m))$ and thus }
    &\leq O(m^{-1})\left(f(\mathbf{x})\right)^2 + O(1/\poly(\norm{\mathbf{x}}_2m))\cdot C_f^2.
\end{align*}
\end{proof}

The ``frequency domain'' estimator above might look daunting but in practice $f$ can usually be approximated with a few main harmonic moments. In the case where $f$ cannot be approximated by sparse harmonics, we provide the following alternative estimator, which can be computed with at most $O(m^{3}\log^3(nMm))$ evaluations of the function $f_{\geq \zeta}$.
\begin{corollary}[``time-domain'' estimator]
With the setup in \cref{thm:frequency_combination}, define  $\psi_{f;a,b}$ as
\begin{align*}
    \psi_{f;a,b} &= \frac{1}{2}\left(\sigma^2+\int_0^\zeta \gamma^2\,\nu(d\gamma)\right)\frac{\sum_{j=1}^3\sum_{k=0}^{m-1}(X_k^{(j)})^2e^{k/m}}{3m}\\
    &\quad + {m^{-3}(-\Gamma(-1/3))^{-3}}\sum_{j=a}^{b-1}\sum_{k=a}^{b-1}\sum_{l=a}^{b-1}e^{\frac{k+j+l}{3m}}\left(f_{\geq \zeta}(X_j^{(1)})+f_{\geq \zeta}(X_k^{(2)})+f_{\geq \zeta}(X_l^{(3)})\right.
    \\
    &\quad\left. -f_{\geq \zeta}(X_j^{(1)}+X_k^{(2)})-f_{\geq \zeta}(X_k^{(2)}+X_l^{(3)})-f_{\geq \zeta}(X_l^{(3)}+X_j^{(1)})+ f_{\geq \zeta}(X_j^{(1)}+X_k^{(2)}+X_l^{(3)})\right).
\end{align*}
Then $\psi_{f;a,b}$ is a $(1\pm O(m^{-1}))$-approximate of the $f$-moment with negligible bias. In particular,
\begin{align*}
    \psi_{f;a,b} = \Re(\phi_{f;a,b}).
\end{align*}

\end{corollary}
\begin{proof}

Note that
\begin{align*}
    &\prod_{j=1}^3\left(\sum_{k=a}^{b-1}(1-e^{i\gamma X_k^{(j)}})e^{\frac{1}{3}k/m}\right)\\
    &=\sum_{j=a}^{b-1}\sum_{k=a}^{b-1}\sum_{l=a}^{b-1}(1-e^{i\gamma X_j^{(1)}})(1-e^{i\gamma X_k^{(2)}})(1-e^{i\gamma X_l^{(3)}})e^{\frac{1}{3}(k+j+l)/m}
\end{align*}
where for any integers $A_1,A_2,A_3$
\begin{align*}
    &(1-e^{i\gamma A_1})(1-e^{i\gamma A_2})(1-e^{i\gamma A_3}) \\
    &= 1-\sum_j e^{i\gamma A_j}+\sum_{j<k}e^{i\gamma (A_1+A_2)}-e^{i\gamma(A_1+A_2+A_3)}\\
    &= (1-e^{i\gamma(A_1+A_2+A_3)})+\sum (1-e^{i\gamma A_j})-\sum_{j<k}(1-e^{i\gamma (A_j+A_k)}).
\end{align*}
Note that $\Re(1-e^{i\gamma k})=1-\cos(\gamma k)$ for any $k\in \Z$. 
Thus by change the order of the triple-summation and integration we have what we want. Finally, note that for any complex random variable $W$, $\var W = \var \Re W + \var \Im W$ and thus $\var \psi_{f;a,b}\leq \var \phi_{f;a,b}$. 
\end{proof}

\subsection{Signed Combination}
As in \cite{cohen2017hyperloglog,cohen2018stream}, the 
transportation of the mean and the relative variance of $\{V_{\gamma}\}_{\gamma >0}$ in the combined estimator $\int_0^\infty V_\gamma \,\nu(d\gamma)$ relies crucially on the fact that $\nu$ is a positive measure. In general, when $\nu$ is a signed measure, 
the combined mean \cref{eq:cohen_mean} still holds but the combined variance \cref{eq:cohen_var} now depend on $|\nu|$.
\begin{corollary}[signed combination]\label{cor:signed}
Let $f(x)=\sigma_1^2 x^2/2+\int_0^\infty (1-\cos(\gamma x))\,\nu_1(d\gamma)$ and $g(x)=\sigma_2^2 x^2/2+\int_0^\infty (1-\cos(\gamma x))\,\nu_2(d\gamma)$ where $\nu_1,\nu_2$ are positive measures. Let $\phi_{f;a,b}$ and $\phi_{g;a,b}$ be generated from a same truncated triple \SymmetricPoissonTower{} sketch in \cref{thm:frequency_combination}. Then
\begin{align*}
    \E (\phi_{f;a,b}-\phi_{g;a,b}) &= (1+O(m^{-1}))(f(\mathbf{x})-g(\mathbf{x})) + O(\frac{C_f+C_g}{\poly(\norm{\mathbf{x}}_2 m)})\\
    \var (\phi_{f;a,b}-\phi_{g;a,b}) &\leq O(m^{-1}) \cdot \left(f(\mathbf{x})+g(\mathbf{x})\right)^2+ O(\frac{C_f^2+C_g^2}{\poly(\norm{\mathbf{x}}_2 m)}).\\
\end{align*}
\end{corollary}
\begin{proof}
Both follow from \cref{thm:frequency_combination}.
    The mean results from the linearity of expectation. For the variance, note that
    \begin{align*}
         \var (\phi_{f;a,b}-\phi_{g;a,b}) \leq (\sqrt{\var  \phi_{f;a,b}} + \sqrt{\var  \phi_{g;a,b}})^2.
    \end{align*}
\end{proof}
Thus for $(\phi_{f;a,b}-\phi_{g;a,b})$ to be a multiplicative approximation of the $(f-g)$-moment, one needs $f(\mathbf{x})+g(\mathbf{x})=\Omega (f(\mathbf{x})-g(\mathbf{x}))$ for any $\mathbf{x}\in \Z^n$.

\section{Proof of the Main Theorem}\label{sec:main_proof}
The main theorem (\cref{thm:main}) almost directly follows from \cref{thm:frequency_combination} and \cref{cor:signed}. We only need to analyze the space and translate the variance guarantee to $(1\pm \epsilon)$-approximation with Chebyshev's inequality.

\begin{lemma}\label{lem:l1}
With probability $5/6$, all the cells in the triple \SymmetricPoissonTower{} have magnitude at most $O(\poly(Mmn))$. 
\end{lemma}
\begin{proof}
The truncated tower stores $X_j$s from $a=-\Omega(m\log (\norm{\mathbf{x}}_2m))$ to $b=\Omega(m\log (\norm{\mathbf{x}}_2m))$. Suppose $b-a=C_1 m\log (\norm{\mathbf{x}}_2m)$  for some $C_1>0$. Each cell $j$ stores an integer. Specifically $X_j\sim \sum_{v\in[n]}\mathbf{x}_v Y_v$ where $Y_v\sim\SymmetricPoisson(e^{-k/m})$. By a simple coupling argument, we have $|X_j|\leq M Y$ where $M$ is the frequency upper bound and $Y\sim \Poisson(ne^{-k/m})$. Therefore by Markov's inequality,
\begin{align*}
    \pr(|X_j|\geq C_2m\log (\norm{\mathbf{x}}_2m) M ne^{-k/m}) \leq \pr( Y\geq C_2 m\log (\norm{\mathbf{x}}_2m) n e^{-k/m})  \leq \frac{1}{C_2 m\log (\norm{\mathbf{x}}_2m)}.
\end{align*}
Choosing $C_2=18C_1$, by the union bound, with probability $1-\frac{C_1m\log (\norm{\mathbf{x}}_2m)}{C_2 m\log (\norm{\mathbf{x}}_2m)}=17/18$, that $|X_j|< C_2m\log (\norm{\mathbf{x}}_2m) M ne^{-k/m}=O(\poly(Mmn))$ for all $j\in[a,b)$, where we bound $\norm{\mathbf{x}}_2$ by $nM^2$. With probability at least 5/6 that this event occurs for all three truncated towers.    
\end{proof}

\begin{lemma}\label{lem:l2}
    With probability 5/6, the triple \SymmetricPoissonTower{} with $m=\Omega(\epsilon^{-2})$ outputs a correct $(1\pm \epsilon)$-approximation of the $f$-moment.
\end{lemma}
\begin{proof}    
 By \cref{thm:frequency_combination}, an estimate of the $f$-moment can be returned with $O(m^{-1})$ relative variance and an $O(C_f/O(\poly(\norm{\mathbf{x}}_2m)))$ additive bias. Note that for any non-empty stream, the $f$-moment $f(\mathbf{x})$ is at least $\xi(f,M)$. Thus the additive bias  can be absorbed in the multiplicative term since we have $\xi(f,M)=\Omega(C_f/\poly(nMm))$ by assumption in \cref{thm:main}. Therefore, the estimator is a $(1\pm O(m^{-1/2}))$-approximation with probability $5/6$ by Chebyshev's inequality with a suitable constant in $O(m^{-1/2})$. It suffices to set $m=\Omega(\epsilon^{-2})$ for a $(1\pm \epsilon)$-approximation. 
\end{proof}

\begin{proof}[Proof of \cref{thm:main}]
    By a union bound of the events in \cref{lem:l1,lem:l2}, with probability 2/3 the estimate is a correct $(1\pm \epsilon)$-approximation and at the same time $O(\log(Mmn))$ bits suffice to store each cell for the three i.i.d.~copies of truncated tower. 
There are $O(m\log(\norm{\mathbf{x}}_2m))=O(m\log(nMm))$ cells. Thus the total space in bits is $O(m\log^2(nMm))=O(\epsilon^{-2}\log^2(nM\epsilon^{-1}))$. The proof for the signed combination goes the same with the variance guarantee provided by \cref{cor:signed}.
\end{proof}

\section{Conclusion}

The set of tractable $f$-moments are those that can be approximated in polylog space.  It is a major open question to fully characterize this class.  To date, Braverman, Chestnut, Woodruff, and Yang~\cite{braverman2016streaming} have 
a nearly-complete characterization of this class, 
but it cannot handle the class of tractable nearly 
periodic functions.

In this work, we proposed a new method for estimating $f$-moments based on a harmonic decomposition of $f$.  The \SymmetricPoissonTower{} is a new 
$O(\epsilon^{-2}\log^2 n)$-space data structure that
can estimate all the important functions in Braverman et al.'s class~\cite{braverman2016streaming}, as well as many nearly periodic functions.  We conjecture that the \SymmetricPoissonTower{} is a universal sketch for the class of tractable functions. 

\begin{conjecture}    
The \SymmetricPoissonTower{} is universal for the class of tractable functions.
\end{conjecture}

There is some subtlety in this conjecture, as the tractable class probably does not have a sharp boundary.
For fixed $f,g \in \mathcal{F}_*$ (\cref{sec:new-results}), define
$\tilde{h}(M)=(\max_{x\in [M]} \frac{f(x)+g(x)}{f(x)-g(x)})^2$.
As $\tilde{h}(M)$ becomes large, more and more functions that can be written as $f-g$ can be included in the range with the cost of having an increasingly large $\tilde{h}(M)$ factor in space. 
This is analogous to the role that $h(M)$ plays in Braverman et al.~\cite{braverman2016streaming}.




\bibliographystyle{alpha}
\bibliography{refs}

\appendix

\section{Example Functions}
\subsection{Nearly Periodic $g_{np}$}
\label{sec:decomp_gnp}
The function $g_{np}$ is constructed in \cite{braverman2016streaming} for the purpose of showing there are tractable functions outside the reach of $L_2$-heavy hitters framework. We decompose $g_{np}$ into harmonic moments and compute its function constant $C_{g_{np}}$ in this section. Recall the following definitions.
\begin{itemize}
    \item 
$\mathbb{B}=\{\sum_{j=1}^k a_j2^{-j}\mid \forall j, a_j\in\{0,1\},k\in\N\}$ is the set of finite-precision binary numbers in $[0,1)$; 
\item $\tau:\Z_+\to \Z_+$ with $\tau(x)=\max\{j\in\N\mid 2^j|x\}$ which returns the position of the least significant bit of $x$; 
\item  $\tau_*:\mathbb{B}\to \Z_+$ with $\tau_*(r)=\min\{j\in\N\mid 2^{-j}|r\}$ which returns the length of the binary representation of $r$.
\end{itemize}

\begin{lemma}
For any $x\in \N$, 
    \begin{align*}
        g_{np}(x) &=2^{-\tau(x)}=\frac{4}{3} \sum_{\gamma \in \mathbb{B}}(1-\cos(2\pi \gamma x))2^{-2\tau_*(\gamma)}
    \end{align*}
\end{lemma}
\begin{proof}
Partition $\mathbb{B}$ as $B_0,B_1,B_2,\ldots$ where $B_k=\{r\in\mathbb{B} \mid \tau_*(r)=k\}$. We have $B_0=\{0\}$, $B_1=\{1/2\}$, $B_2=\{1/4,3/4\}$, $B_3=\{1/8,3/8,5/8,7/8\}$, and so on. 
By symmetry, for $k\geq 2$
\begin{align*}
   \sum_{\gamma \in B_k}(1-\cos(2\pi \gamma x)) &=\begin{cases}
       2^{k-1}, & \tau(x)\leq k-2  \\
       2^{k}, & \tau(x)=k-1  \\
        0, & \tau(x)\geq k 
   \end{cases}
\end{align*}
where in the first case $\sum_{\gamma \in B_k}\cos(2\pi \gamma x)=0$, in the second case $\cos(2\pi \gamma x)=-1$ for any $\gamma\in B_k$, and in the last case, $\cos(2\pi \gamma x)=1$ for any $\gamma\in B_k$.
Thus if $\tau(x)\geq 1$, i.e., $x$ is even, then
\begin{align*}
    \sum_{\gamma \in \mathbb{B}}(1-\cos(2\pi \gamma x))2^{-2\tau_*(\gamma)} & = \sum_{k\in \N} \sum_{\gamma \in B_k}(1-\cos(2\pi \gamma x))2^{-2k} \\
    &= 2^{\tau(x)+1} 2^{-2(\tau(x)+1)}  +\sum_{k=\tau(x)+2}^{\infty}2^{k-1}\cdot  2^{-2k}\\
    &= 2^{-(\tau(x)+1)}+2^{-(\tau(x)+2)}\\
    &= \frac{3}{4}2^{-\tau(x)}.
\end{align*}
When $x$ is odd, then 
\begin{align*}
    \sum_{\gamma \in \mathbb{B}}(1-\cos(2\pi \gamma x))2^{-2\tau_*(\gamma)} & = \sum_{k\in \N} \sum_{\gamma \in B_k}(1-\cos(2\pi \gamma x))2^{-2k} \\
    &= 2 \cdot 2^{-2}+\sum_{k=2}^\infty 2^{k-1} \cdot 2^{-2k}\\
    &= \frac{3}{4}.
\end{align*}
\end{proof}
\begin{lemma}   
$C_{g_{np}}\leq \frac{2}{3}$.
\end{lemma}
\begin{proof}
\begin{align*}
  C_{g_{np}} &=   \frac{4}{3}\sum_{\gamma \in \mathbb{B}}\min(4\pi^2\gamma^2,1)2^{-2\tau_*(\gamma)} \leq \frac{4}{3}\sum_{k=1}^\infty \sum_{\gamma \in B_k}2^{-2k} =\frac{4}{3}\sum_{k=1}^\infty 2^{k-1} 2^{-2k} = \frac{2}{3}.
\end{align*}
    
\end{proof}

\subsection{The Golden Harmonic Moment $g_{gold}$}\label{sec:gold}
The golden harmonic moment $g_{gold}$ is constructed for additional examples that are tractable in the harmonic sketching framework but out of the reach of prior techniques.

With the pigeonhole argument in \cref{lem:why_fail}, it is clear that  $\xi(g_{gold},M) =O(1/M^2)$. We now bound $\xi(g_{gold},M)$ from below, which is based on the theory of continued fractions. 
\begin{lemma} 
Let $g_{gold}(x) = 1-\cos(\frac{1+\sqrt{5}}{2}\cdot 2\pi x)$ for $x\in \Z$. Then $\xi(g_{gold},M) =\Omega(1/M^2)$. 
\end{lemma}
\begin{proof}
    The convergents of $\frac{1+\sqrt{5}}{2}$ are $(F_{n+1}/F_{n})_{n\in \N}$ where $(F_n)_{n\in \N}$ is the Fibonacci sequence. It is known that for any $p,q\in \N$ such that $q\leq F_n$, $|\frac{1+\sqrt{5}}{2}-p/q|\geq |\frac{1+\sqrt{5}}{2}-F_{n+1}/F_n|$. Thus for any $x\in \N$, let $\frac{1+\sqrt{5}}{2} x = w +r $ where $w\in \N$ and $r\in (-1/2,1/2)$. Let $F_n\geq x$. Then
    \begin{align*}
        |r/x| = |\frac{1+\sqrt{5}}{2}-\frac{w}{x}| \geq |\frac{1+\sqrt{5}}{2}-F_{n+1}/F_n| \geq |\frac{1+\sqrt{5}}{2}-F_{n+2}/F_{n+1}|
    \end{align*}
    Thus 
    \begin{align*}
        2|r/x| &\geq |\frac{1+\sqrt{5}}{2}-F_{n+1}/F_n| + |\frac{1+\sqrt{5}}{2}-F_{n+2}/F_{n+1}|\\
        &\geq |F_{n+1}/F_n-F_{n+2}/F_{n+1}|\\
        &= \frac{1}{F_n F_{n+1}}
    \end{align*}
    Note that $F_n = \Theta((\frac{1+\sqrt{5}}{2})^n)$ and thus we may choose $F_n,F_{n+1}$ that are $O(x)$. Therefore $|r|=\Omega(1/x)$, which implies $1-\cos(\frac{1+\sqrt{5}}{2}\cdot 2\pi x)=1-\cos(w2\pi+2\pi r)=1-\cos(2\pi r)\leq 2\pi^2 r^2=\Omega(1/x^2)$.
\end{proof}

\section{Mathematical Lemmas}\label{sec:math-lemmas}
In this section we estimate the rate that sums like $\frac{1}{m}\sum_{k\in\Z}(1-e^{-e^{-k/m}})e^{\frac{1}{3}k/m}$ converge to its limit $\int_{-\infty}^\infty(1-e^{-e^{-r}})e^{\frac{1}{3}r}\,dr$ as $m$ increases, supporting our analysis of the smoothed (-1/3)-aggregation.
\begin{lemma}\label{lem:sum_to_integral}
    Let $h:\R\to \C$ be a differentiable function. If both $h$ and $h'$ are (Lebesgue) integrable, then 
    \begin{align*}
        \left|\int_{-\infty}^\infty h(s)\,ds - \frac{1}{m}\sum_{k=-\infty}^\infty h(k/m)\right| \leq m^{-1} \int_{-\infty}^\infty |h'(s)|\,ds.
    \end{align*}
\end{lemma}
\begin{proof}
    We bound the difference    
\begin{align*}
    \left|\int_{-\infty}^\infty h(s)\,ds - \frac{1}{m}\sum_{k=-\infty}^\infty h(k/m)\right| &= \left|\sum_{k=-\infty}^\infty \left(\int_0^{1/m}h(k/m+s)\,ds - \frac{1}{m}h(k/m)\right)\right|\\
        &\leq \sum_{k=-\infty}^\infty \left|\int_0^{1/m}h(k/m+s)\,ds - \frac{1}{m}h(k/m)\right|.
\end{align*}
Note that
\begin{align*}
\left|\int_0^{1/m}h(k/m+s)\,ds - \frac{1}{m}h(k/m)\right|&= \left|\int_0^{1/m}\int_0^s h'(k/m+t)\,dtds\right|\\
&\leq \int_0^{1/m}\int_0^s |h'(k/m+t)|\,dtds\\
&\leq \int_0^{1/m}\int_0^{1/m} |h'(k/m+t)|\,dtds\\
&=\frac{1}{m}\int_0^{1/m} |h'(k/m+t)|\,dt.
\end{align*}
Thus
\begin{align*}
    \left|\int_{-\infty}^\infty h(s)\,ds - \frac{1}{m}\sum_{j=-\infty}^\infty h(j/m)\right|
    &\leq\sum_{k=-\infty}^\infty \frac{1}{m}\int_0^{1/m} |h'(k/m+t)|\,dt\\
    &= \frac{1}{m}\int_{-\infty}^{\infty} |h'(t)|\,dt.
\end{align*}
\end{proof}

\begin{lemma}\label{lem:eta}
    Let $a,b\in \C$ and $c\in \R$. Define 
    \begin{align*}
        \eta_1(x;a,b,c) &= (e^{-ae^{-x}}-e^{-be^{-x}})e^{cx},
&        \eta_2(x;a,c) &= ae^{-x}e^{-ae^{-x}}e^{cx}.
    \end{align*} 
    
    If $\Re(a)\geq 0$, $\Re(b)\geq 0$ and $c\in (0,1)$, then
    \begin{enumerate}
        \item $\int_{-\infty}^\infty \eta_1(x;a,b,c)\,dx = (a^c-b^c)\Gamma(-c),$ where $\Gamma(\cdot)$ here denotes the Gamma function (\cref{fig:gamma_function});\label{item:lem:eta1}
        \item $\eta_1'(x;a,b,c)=c\eta_1(x;a,b,c)+\eta_2(x;a,c)-\eta_2(x;b,c)$;\label{item:lem:eta2}
        \item $\int_{-\infty}^\infty |\eta_1(x;a,b,c)|\,dx \leq  2\left(\frac{1}{c}+\frac{1}{1-c}\right)|b-a|^c$.\label{item:lem:eta3}
        \item $\int_{-\infty}^\infty |\eta_2(x;a,c)|\,dx \leq |a|^{c} \Gamma(1 - c)$.\label{item:lem:eta4}
        \item $\int_{-\infty}^\infty |\eta_1'(x;a,b,c)|\,dx\leq 2\left(1+\frac{c}{1-c}\right)|b-a|^c+(|a|^c+|b|^c)\Gamma(1-c)$.\label{item:lem:eta5}
        \item \begin{align*}
        \lefteqn{\left|\int_{-\infty}^\infty \eta_1(x;a,b,c)\,dx-\frac{1}{m}\sum_{k=-\infty}^{\infty}\eta_1(k/m;a,b, c)\right|}\\
        &{\hspace{4cm}} \leq m^{-1}\left(2\left(1+\frac{c}{1-c}\right)|b-a|^c+(|a|^c+|b|^c)\Gamma(1-c)\right).   
        \end{align*}\label{item:lem:eta6}
    \end{enumerate}
\end{lemma}
\begin{proof}
\underline{Part 1.} 
Integral by parts.
    \begin{align*}
        &\int_{-\infty}^\infty (e^{-ae^{-x}}-e^{-be^{-x}})e^{cx}\,dx \\
        &= \left. (e^{-ae^{-x}}-e^{-be^{-x}})c^{-1}e^{cx}\right|_{-\infty}^\infty - \int_{-\infty}^\infty (ae^{-x}e^{-ae^{-x}}-be^{-x}e^{-be^{-x}})c^{-1}e^{cx}\,dx
        \intertext{Since $\Re(a)\geq 0$, $\Re(b)\geq 0$ and $c\in (0,1)$,  $\left. (e^{-ae^{-x}}-e^{-be^{-x}})c^{-1}e^{cx}\right|_{-\infty}^\infty=0$,}
        &= -c^{-1} \int_{-\infty}^\infty ae^{-(1-c)x}e^{-ae^{-x}}\,dx + c^{-1} \int_{-\infty}^\infty be^{-(1-c)x}e^{-be^{-x}}\,dx.
    \end{align*}
    Set $z=ae^{-x}$ and we have
    \begin{align}
        c^{-1}\int_{-\infty}^\infty ae^{-(1-c)x}e^{-ae^{-x}}\,dx &=c^{-1}\int_{0}^\infty a^{c}z^{-c}e^{-z}\,dz\nonumber\\
        &= a^c \frac{\Gamma(1-c)}{c}=a^c (-\Gamma(-c)).\label{eq:acgamma}
    \end{align}
    Thus
    \begin{align*}
        \int_{-\infty}^\infty (e^{-ae^{-x}}-e^{-be^{-x}})e^{cx}\,dx&=-a^c (-\Gamma(-c))+b^c (-\Gamma(-c))=(a^c-b^c)\Gamma(-c).
    \end{align*}

\medskip
\noindent\underline{Part 2.} 
    \begin{align*}
        \eta_1'(x;a,b,c)&=(e^{-ae^{-x}}-e^{-be^{-x}})ce^{cx}+(ae^{-x}e^{-ae^{-x}}-be^{-x}e^{-be^{-x}})e^{cx}\\
        &= c\eta_1(x;a,b,c)+\eta_2(x;a,c)-\eta_2(x;b,c).
    \end{align*}

\medskip
\noindent\underline{Part 3.} 
    Define the path $\phi(t)=(1-t)a+tb$ and $g_x(r)=e^{-re^{-x}}$. By the path integral,
    \begin{align*}
        \left| e^{-ae^{-x}}-e^{-be^{-x}}\right| &=  |g_x(a)-g_x(b)|=\left|\int_0^1 g_x'(\phi(t))\phi'(t)\,dt\right|\\
        &=\left| \int_0^1 e^{-\phi(t)e^{-x}}(-e^{-x})(b-a)\,dt\right|\leq  |b-a|e^{-x}\int_0^1 |e^{-\phi(t)e^{-x}}|\,dt
        \intertext{We assumed $\Re(a),\Re(b)\geq 0$ and thus $\Re(\phi(t))\geq 0$ for any $t\in[0,1]$, which implies $|e^{-\phi(t)e^{-x}}|<1$. We then have }
        \left| e^{-ae^{-x}}-e^{-be^{-x}}\right|&\leq |b-a|e^{-x}.
    \end{align*}
    On the other hand, since $\Re(a),\Re(b)\geq 0$, we have $\left|e^{-ae^{-x}}\right|,\left|e^{-be^{-x}}\right|\leq 1$ and thus
    \begin{align*}
        \left| e^{-ae^{-x}}-e^{-be^{-x}}\right| &\leq \left|e^{-ae^{-x}}\right|+\left|e^{-ae^{-x}}\right|\leq 2.
    \end{align*}
    We conclude that for any $q\in \R$
    \begin{align*}
        \int_{-\infty}^\infty \left|(e^{-ae^{-x}}-e^{-be^{-x}})e^{cx}\right| \,dx &=\int_{-\infty}^q \left|e^{-ae^{-x}}-e^{-be^{-x}}\right|e^{cx} \,dx   +\int_{q}^{\infty} \left|e^{-ae^{-x}}-e^{-be^{-x}}\right|e^{cx} \,dx  \\
        &\leq \int_{-\infty}^q 2 e^{cx} \,dx   +\int_{q}^{\infty} |b-a|e^{-x}e^{cx}\,dx\\
        &=2 \frac{e^{cq}}{c}+|b-a|\frac{e^{(c-1)q}}{1-c},
        \intertext{since $c\in(0,1)$. Now set $q=\log |b-a|$ and we have }
        \int_{-\infty}^\infty \left|(e^{-ae^{-x}}-e^{-be^{-x}})e^{cx}\right| \,dx &\leq 2\left(\frac{1}{c}+\frac{1}{1-c}\right)|b-a|^c.
    \end{align*}

\medskip
\noindent\underline{Part 4.} 
    \begin{align*}
        \int_{-\infty}^\infty \left|a e^{-x}e^{-ae^{-x}}e^{cx}\right|\, dx &= |a|\int_{-\infty}^\infty  e^{-x}e^{-\Re({a})e^{-x}}e^{cx}\, dx
        \intertext{by \cref{eq:acgamma}, since $\Re({a})\geq 0$}
         &= |a|\frac{1}{\Re(a)} (\Re(a))^c \Gamma(1-c)\\
         &\leq |a|^c\Gamma(1-c).
    \end{align*}
\medskip
\noindent\underline{Part 5.} 
Follows from 2, 3, and 4.

\medskip
\noindent\underline{Part 6.}  Follows from 5 and \cref{lem:sum_to_integral}.
\end{proof}

\end{document}